\documentclass[12pt]{article}
\usepackage{amsmath}
\usepackage{graphicx,psfrag,epsf}
\usepackage{enumerate}
\usepackage[authoryear, sort]{natbib}
\usepackage{url} 

\newcommand{\blind}{1}

\usepackage{amssymb}
\usepackage{graphicx}
\usepackage{theorem}
\usepackage{caption}
\usepackage{subcaption}
\newenvironment{proof}{{\bf Proof:  }}{\hfill\rule{2mm}{2mm}}
\newtheorem{fact}{Fact}[section]
\newtheorem{lemma}[fact]{Lemma}
\newtheorem{theorem}[fact]{Theorem}
\newtheorem{definition}[fact]{Definition}

\newtheorem{remark}[fact]{Remark}
\usepackage{tikz}
\usetikzlibrary{arrows}
\usepackage{algorithm}
\usepackage{algorithmic}
\usepackage{bm}

\DeclareMathOperator*{\argmax}{arg\,max}

\renewcommand{\hat}{\widehat}
\renewcommand{\tilde}{\widetilde}
\usepackage{nameref, hyperref}

\begin{document}

\def\spacingset#1{\renewcommand{\baselinestretch}%
{#1}\small\normalsize} \spacingset{1}


\if1\blind
{
  \title{\bf Robust Estimation from Multiple Graphs\\under Gross Error Contamination}
  \author{Runze Tang, Minh Tang, Joshua T.\ Vogelstein, Carey E.\ Priebe\thanks{
    The authors gratefully acknowledge \textit{the D3M program of the Defense Advanced Research Projects Agency (DARPA) administered through AFRL contract FA8750-17-2-0112; DARPA XDATA contract FA8750-12-2-0303; DARPA SIMPLEX contract N66001-15-C-4041; and DARPA GRAPHS contract N66001-14-1-4028. The authors also wish to thank Greg Kiar for helpful support on brain graphs data.}}\hspace{.2cm}\\
    Department of Applied Maths and Statistics, Johns Hopkins University}
  \maketitle
} \fi

\if0\blind
{
  \bigskip
  \bigskip
  \bigskip
  \begin{center}
    {\LARGE\bf Robust Estimation from Multiple Graphs\\under Gross Error Contamination}
\end{center}
  \medskip
} \fi

\bigskip
\begin{abstract}
Estimation of graph parameters based on a collection of graphs is essential for a wide range of graph inference tasks. In practice, weighted graphs are generally observed with edge contamination. We consider a weighted latent position graph model contaminated via an edge weight gross error model and propose an estimation methodology based on robust L$q$ estimation followed by low-rank adjacency spectral decomposition. We demonstrate that, under appropriate conditions, our estimator both maintains L$q$ robustness and wins the bias-variance tradeoff by exploiting low-rank graph structure. We illustrate the improvement offered by our estimator via both simulations and a human connectome data experiment.
\end{abstract}

\noindent%
\vfill

\newpage
\spacingset{1.45} 

\section{Background and Overview}

Network analysis has emerged as an area of intense statistical theory and application activity. In the general parametric framework, $G \sim f \in \mathcal{F} = \{f_{\theta} : \theta \in \Theta \}$, and selecting a principled and productive estimator $\hat{\theta}$ for the unknown graph parameter $\theta$ given a sample of graphs $\{G^{(1)}, \cdots, G^{(m)}\}$ is one of the most foundational and essential tasks, facilitating subsequent inference.
For example, \citet{ginestet2014hypothesis} proposes a method to test for a difference between the networks of two groups of subjects in functional neuroimaging; while hypothesis testing is the ultimate goal, estimation is a key intermediate step.
We propose a widely-applicable, robust, low-rank estimation procedure for a collection of weighted graphs.

Consider for illustration the connectome dataset ``Templeton114'' made available through the
Neurodata repository\footnote{\url{http://m2g.io/}}
and investigated in Section~\ref{section:real_data} below.
We have $m=114$ brain graphs, each having $n=70$ vertices representing different anatomical regions; the (errorfully observed) weight of an edge between two vertices represents the number of white-matter tracts connecting the corresponding two regions of the brain, as measured by
diffusion tensor magnetic resonance imaging.
Our goal in this situation is to estimate the average number of white-matter tracts between different regions of the brain. A more accurate estimate can lead to a better understanding of brain connectivity and hence functionality. Also, better estimates will improve performance on other tasks, such as diagnosis of brain disease. 

The maximum likelihood estimate (MLE) -- the edge-wise sample mean, without taking any graph structure into account, as in the (weighted extension of) the independent edge graph model (IEM) \citep{bollobas2007phase} (described in Section~\ref{section:WIEM} below) --
is a natural candidate for our estimation problem. However, the MLE suffers from at least two major deficiencies in our setting: high variance and non-robustness.

In our high dimensional setting (a large number of vertices, $n$), the edge-wise MLE leads to  estimates with unacceptably high variance unless the sample size (the number of graphs, $m$) is exceedingly large.
However, if the graphs can be assumed to be (approximately) low-rank, then by biasing towards low-rank structure, more elaborate estimators can have greatly reduced variance and win the bias-variance tradeoff.
For our connectome data (Section~\ref{section:real_data} Figure~\ref{fig:screehist})
we observe this approximate low-rank property. \citet{tang2016law} develops an estimator based on a low-rank approximation and proves that this new estimator outperforms the edge-wise MLE, decreasing the overall  variance dramatically by smoothing towards the low-rank structure.

The second edge-wise MLE deficiency in our setting derives from the edge observations being subject to contamination. That is, the weights attributed to edges are possibly observed with noise.
The sample mean is notoriously un-robust to outliers;
thus, under the possibility of contamination, it is wise to use robust methods, such as the ML$q$E \citep{ferrari2010maximum, qin2013maximum} considered in this paper.

To address these two deficiencies simultaneously, we propose an estimation methodology which is a natural extension of \citep{tang2016law} to gross error contamination. Our proposed estimator both inherits ML$q$E robustness and wins the bias-variance tradeoff by taking advantage of low-rank structure.

We organize the paper as follows. In Section~\ref{section:model}, we extend the independent edge model, random dot product graph model, and stochastic blockmodel to the weighted versions, and define the gross error contamination model we will consider. In Section~\ref{section:estimators}, we present our estimation methodology in terms of two estimators
designed to address the two edge-wise MLE deficiencies described above, and we construct our final estimator by combining the two estimators. In Section~\ref{section:theory}, we prove that our estimator is superior, under appropriate conditions, and this result is generalized in Section~\ref{section:extension}. In Section~\ref{section:results}, we illustrate the performance of our estimator through experimental results on simulated and real data.

\section{Models}
\label{section:model}

For this work, we are interested in the scenario where $m$ weighted graphs on $n$ vertices are given as adjacency matrices $\{ \bm{A}^{(t)} \} (t = 1, \dotsc, m)$. The graphs are undirected without self-loops, i.e.\ each $\bm{A}^{(t)}$ is symmetric with zeros along the diagonal. Moreover, we assume the vertex correspondence is known across different graphs, so that vertex $i$ of the $t_1$-th graph corresponds to vertex $i$ of the $t_2$-th graph for any $i \in [n]$ and $t_1, t_2 \in [m]$.

In this section, we present three nested models, the weighted independent edge model (WIEM) in Section~\ref{section:WIEM}, the weighted random dot product graph model (WRDPG) in Section~\ref{section:WRDPG}, and the weighted stochastic blockmodel (WSBM) as a WRDPG in Section~\ref{section:WSBM}. Moreover, we introduce a contaminated model based on Section~\ref{section:WSBM} in Section~\ref{section:Contamination}.

\subsection{Weighted Independent Edge Model}
\label{section:WIEM}

In an independent edge model (IEM) \citep{bollobas2007phase} with probability matrix $\bm{P} \in [0, 1]^{n \times n}$, every edge weight $\bm{A}_{ij}$ is drawn from a Bernoulli distribution with parameter $\bm{P}_{ij}$ independent of all other edges.
We first extend the definition of IEM to the weighted independent edge model (WIEM) with respect to a one-parameter family $\mathcal{F} = \{ f_{\theta} : \theta \in \Theta \subset \mathbb{R} \}$; for example, $f_{\theta}$ may be the Poisson distribution with parameter $\theta$. Denote the graph parameters as a matrix $\bm{P} \in \Theta^{n \times n} \subset \mathbb{R}^{n \times n}$. Then under a WIEM, the (weighted) edge between vertex $i$ and vertex $j$ ($i < j$ due to symmetry) has weight $\bm{A}_{ij}$ drawn from $f_{\bm{P}_{ij}}$ independent of all other edges.
Thus IEM is a special case of WIEM, with $\mathcal{F}$ representing the collection of Bernoulli distributions and $\Theta = [0, 1]$.

Note that the graphs considered in this paper are undirected without self-loops, and the parameter matrix $\bm{P}$ can be considered to be symmetric and hollow. That is, for convenience, we still define the parameters to be an $n$-by-$n$ matrix while only $n*(n-1)/2$ of them are active.

\subsection{Weighted Random Dot Product Graph}
\label{section:WRDPG}

Vertices are selective about their adjacencies in graphs. A vertex may be frequently adjacent to one group of vertices but rarely adjacent to the other group of vertices.
The latent position model proposed by \citet{hoff2002latent} captures such properties and model these differences in vertex properties by assigning to each vertex $i$ a corresponding latent vector $\bm{X}_i \in \mathbb{R}^d$. Conditioned on the latent vectors $\bm{X}_i$ and $\bm{X}_j$, the edge weight between vertex $i$ and vertex $j$ is independent of all other edges and depends only on $\bm{X}_i$ and $\bm{X}_j$ through a link function.

A special case of the latent position model is the random dot product graph model (RDPG) in which the link function is the inner product \citep{nickel2008random, young2007random}.
Now we give a definition of the weighted random dot product graph (WRDPG) as a special case of the weighted latent position model as follows:
\begin{definition}[Weighted Random Dot Product Graph Model]
Consider a collection of one-parameter distributions $\mathcal{F} = \{ f_{\theta}, \theta \in \Theta \subset \mathbb{R} \}$. The weighted random dot product graph model (WRDPG) with respect to $\mathcal{F}$ is defined via consideration of latent position matrix $\bm{X} \in \mathbb{R}^{n \times d}$ such that $\bm{X} = [\bm{X}_1, \bm{X}_2, \dotsc, \bm{X}_n]^{\top}$, where $\bm{X}_i \in \mathbb{R}^d$ for all $i \in [n]$. The matrix $\bm{X}$ is random and satisfies $\mathbb{P}\left[ \bm{X}_i^{\top} \bm{X}_j \in \Theta \right] = 1$ for all $i, j \in [n]$. Conditioned on $\bm{X}$, the entries of the adjacency matrix $\bm{A}$ are independent and $\bm{A}_{ij}$ is a random variable following distribution $f_{\theta} \in \mathcal{F}$ with parameter $\theta = \bm{X}_i^{\top} \bm{X}_j $ for all $i < j \in [n]$.
\end{definition}
Under the WRDPG defined above, the parameter matrix $\bm{P} = \bm{X} \bm{X}^{\top} \in \Theta^{n \times n} \subset \mathbb{R}^{n \times n}$ is automatically symmetric because the link function is the inner product. Moreover, to have symmetric graphs without self-loops, only $\bm{A}_{ij}$ ($i < j$) are sampled while leaving the diagonals of $\bm{A}$ to be all zeros.


\subsection{Weighted Stochastic Blockmodel as a Weighted Random Dot Product Graph}
\label{section:WSBM}

Community structure is an important property of graphs under which vertices are clustered into different communities such that vertices within the same community behave similarly. The stochastic blockmodel (SBM) proposed by \citet{holland1983stochastic} captures such a property, where each vertex is assigned to one block and the connectivity between two vertices depends only on their respective block memberships.

Formally, the SBM is determined by the number of blocks $K$ (generally much smaller than the number of vertices $n$), block probability matrix $\bm{B} \in [0, 1]^{K \times K}$, and the block assignment vector $\bm{\tau} \in [K]^n$, where $\bm{\tau}_i = k$ represents that vertex $i$ belongs to block $k$. Conditioned on the block membership $\bm{\tau}$, the connectivity between vertex $i$ and vertex $j$ follows a Bernoulli distribution with parameter $\bm{B}_{\bm{\tau}_i, \bm{\tau}_j}$. This can be easily generalized to the weighted stochastic blockmodel (WSBM), with the Bernoulli distribution replaced by a one-parameter distribution family $\mathcal{F} = \{ f_{\theta}, \theta \in \Theta \subset \mathbb{R} \}$ and the block probability matrix given by $\bm{B} \in \Theta^{K \times K} \subset \mathbb{R}^{K \times K}$.

Since the RDPG/WRDPG setting motivates low-rank estimation -- $\bm{P}$ is of rank less than or equal to $d$ -- all analysis in this work is based on such a setting. In order to consider WSBM as a WRDPG, the block probability matrix $\bm{B}$ needs to be positive semi-definite by the structure of WRDPG. Henceforth, with slight abuse of terminology, we will denote the sub-model of WSBM with positive semi-definite $\bm{B}$ as simply the WSBM.

Now consider the WSBM as a WRDPG with respect to $\mathcal{F} = \{ f_{\theta}, \theta \in \Theta \subset \mathbb{R} \}$. Letting $d =
\mathrm{rank}(\bm{B})$, all vertices in block $k$ have shared latent position $\bm{\nu}_k \in \mathbb{R}^{d}$, where $\bm{B} = \bm{\nu} \bm{\nu}^{\top}$ and $\bm{\nu} = [\bm{\nu}_1, \dotsc, \bm{\nu}_K]^{\top} \in \mathbb{R}^{K \times d}$.
That is to say, $\bm{X}_i = \bm{\nu}_{\bm{\tau}_i}$ and $\bm{A}_{ij}$ $(i < j)$ is distributed as $f$ with parameter $\bm{B}_{\bm{\tau}_i, \bm{\tau}_j} = \bm{\nu}_{\bm{\tau}_i}^{\top} \bm{\nu}_{\bm{\tau}_j}$. Here the parameter matrix $\bm{P} \in \mathbb{R}^{n \times n}$ is symmetric and satisfies $\bm{P}_{ij} = \bm{X}_i^{\top} \bm{X}_j = \bm{\nu}_{\bm{\tau}_i}^{\top} \bm{\nu}_{\bm{\tau}_j} = \bm{B}_{\bm{\tau}_i, \bm{\tau}_j}$.

In order to generate $m$ graphs under this model with known vertex correspondence, we first sample $\bm{\tau}$ from the categorical distribution with parameter $\bm{\rho} = [\bm{\rho}_1, \cdots, \bm{\rho}_K]^{\top}$ with $\bm{\rho}_k \in (0, 1)$ and $\sum_{k = 1}^{K} \bm{\rho}_k = 1$, and keep $\bm{\tau}$ fixed when sampling all $m$ graphs. Then $m$ symmetric and hollow graphs are sampled such that conditioning on $\bm{\tau}$, the adjacency matrices are distributed entry-wise independently as $\bm{A}^{(t)}_{ij} \stackrel{ind}{\sim} f_{\bm{B}_{\bm{\tau}_i, \bm{\tau}_j}} = f_{\bm{P}_{ij}}$ for each $1 \le t \le m$, $1 \le i < j \le n$.

\subsection{Weighted Stochastic Blockmodel as a Weighted Random Dot Product Graph with Contamination}
\label{section:Contamination}

In practice, completely accurate data is difficult to collect -- there will almost always be noise in the observations which deviates from our general model assumptions. In order to incorporate this effect, the gross error contamination model \citep{bickel2007mathematical} is considered in this work.

Generally in a gross error model, we observe good measurement $G^* \sim f_{\bm{P}} \in \mathcal{F}$ most of the time, while there are a few contaminated values $G^{**} \sim h_{\bm{C}} \in \mathcal{H}$ when gross errors occur. Here $\bm{P}$ and $\bm{C}$ represent the respective parameter matrices of the two distribution families.
As for graphs, one way to generalize to the gross error model is to contaminate the entire graph with some small probability $\epsilon \in (0, 1)$, that is $G \sim (1-\epsilon) f_{\bm{P}} + \epsilon h_{\bm{C}}$. However, since all the models we consider are subsets of the WIEM, it is more natural to consider the contamination with respect to each edge, i.e.\ for $1 \le i <  j \le n$, $G_{ij} \sim (1-\epsilon) f_{\bm{P}_{ij}} + \epsilon h_{\bm{C}_{ij}}$ with $f \in \mathcal{F}$ and $h \in \mathcal{H}$, where both $\mathcal{F}$ and $\mathcal{H}$ are one-parameter distribution families.

In this paper, we assume that when gross errors occur, the weights of the edges are also from the same one-parameter family $\mathcal{F}$. Moreover, we also assume that the connectivity follows the WSBM as a WRDPG. Thus, similar to the uncontaminated distribution $f_{P_{ij}}$ with $\bm{P}_{ij} = \bm{B}_{\bm{\tau}_i, \bm{\tau}_j}$ where $\bm{B}$ is the block probability matrix and $\bm{\tau}$ is the block assignments, the contamination distribution $f_{\bm{C}_{ij}}$ with $\bm{C}_{ij} = \bm{B}^{\prime}_{\bm{\tau}^{\prime}_i, \bm{\tau}^{\prime}_j}$ also has the block structure, where $\bm{B}^{\prime}$ is the block probability matrix and $\bm{\tau}^{\prime}$ is the block assignment vector. For clarity, we will consider the sampling procedure when the contamination has the same block structure, i.e.\ $\bm{\tau} = \bm{\tau}^{\prime}$. However, this simplification is not required in our theory.

To generate $m$ graphs under this contamination model with known vertex correspondence, we first sample $\bm{\tau}$ from the categorical distribution with parameter $\bm{\rho}$ and keep $\bm{\tau}$ fixed for all $m$ graphs as in Section \ref{section:WSBM}. Then $m$ symmetric and hollow graphs $G^{(1)}, \dotsc, G^{(m)}$ are sampled such that conditioning on $\bm{\tau}$, the adjacency matrices are distributed entry-wise independently as $\bm{A}^{(t)}_{ij} \stackrel{ind}{\sim} (1-\epsilon) f_{\bm{P}_{ij}} + \epsilon f_{\bm{C}_{ij}}$ for each $1 \le t \le m$, $1 \le i < j \le n$, where $\bm{P}_{ij} = \bm{B}_{\bm{\tau}_i, \bm{\tau}_j}$ and $\bm{C}_{ij} = \bm{B}^{\prime}_{\bm{\tau}_i, \bm{\tau}_j}$. Here $\epsilon$ is the probability of an edge to be contaminated, $\bm{P}$ is the parameter matrix as in Section \ref{section:WSBM}, and $\bm{C}$ is the parameter matrix for contamination.

\section{Estimators}
\label{section:estimators}
Under any model introduced in Section~\ref{section:model}, our goal is to estimate the parameter matrix $\bm{P}$ based on the $m$ observations $\bm{A}^{(1)}, \dotsc, \bm{A}^{(m)}$. Especially when under the contamination model, although there are other parameters such as $\epsilon$ and $\bm{C}$, our goal is still to estimate the uncontaminated parameter matrix $\bm{P}$. In this section, we present four estimators as depicted in Figure~\ref{fig:roadmap}, i.e.\ the standard entry-wise MLE $\bm{\hat{P}}^{(1)}$, the low-rank approximation of the entry-wise MLE $\bm{\widetilde{P}}^{(1)}$, the entry-wise robust estimator ML$q$E $\bm{\hat{P}}^{(q)}$, and the low-rank approximation of the entry-wise ML$q$E $\bm{\widetilde{P}}^{(q)}$. Since the observed graphs are symmetric and hollow with a symmetric parameter matrix of the model, we are not concerned with estimating the diagonal of $\bm{P}$; however, the estimate itself should be at least symmetric.

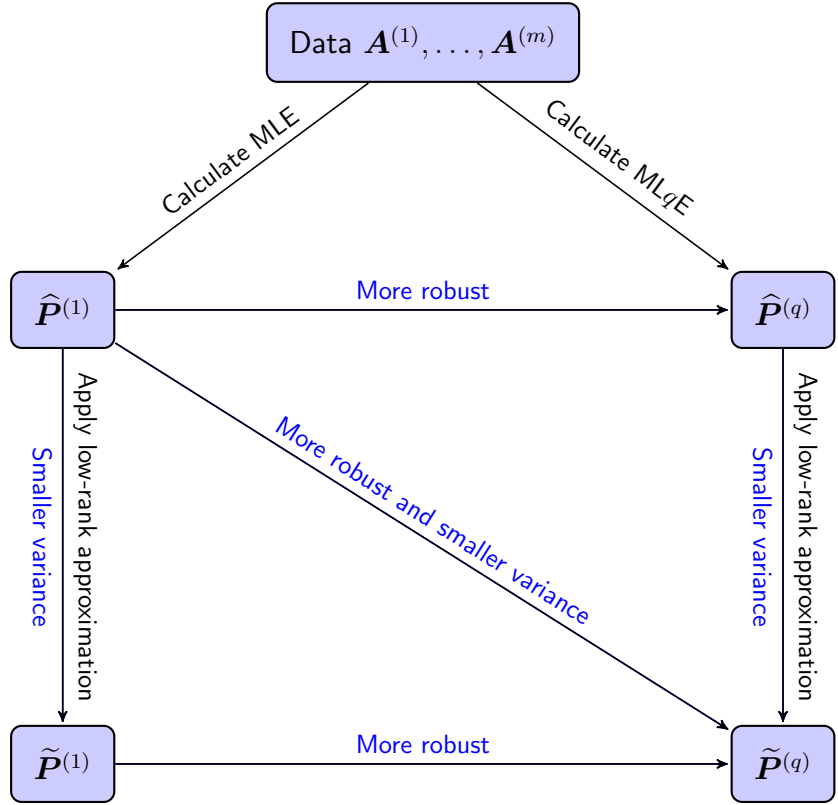
\begin{figure}
\begin{center}
\hspace*{-0.2in}
\begin{tikzpicture}[
  font=\sffamily,
  every matrix/.style={ampersand replacement=\&,column sep=2cm,row sep=2.5cm},
  block/.style={draw,thick,rounded corners,fill=blue!20,inner sep=.3cm},
  process/.style={draw,thick,circle,fill=blue!20},
  sink/.style={source,fill=green!20},
  datastore/.style={draw,very thick,shape=datastore,inner sep=.3cm},
  dots/.style={gray,scale=2},
  to/.style={->,>=stealth',shorten >=1pt,semithick,font=\sffamily\footnotesize},
  tofrom/.style={<->,>=stealth',shorten >=1pt,semithick,font=\sffamily\footnotesize},
  every node/.style={align=center}]

  \matrix{
  	\& \node[block] (Data) {Data $\bm{A}^{(1)}, \dotsc, \bm{A}^{(m)}$};\\
    \node[block] (MLE) {$\hat{\bm{P}}^{(1)}$};
      \& \& \node[block] (MLqE) {$\hat{\bm{P}}^{(q)}$};\\
	\\
    \node[block] (XMLE) {$\widetilde{\bm{P}}^{(1)}$};
      \& \& \node[block] (XMLqE) {$\widetilde{\bm{P}}^{(q)}$}; \\
  };

  \draw[to] (Data) -- node[midway, sloped, above] {Calculate MLE} (MLE);
  \draw[to] (Data) -- node[midway, sloped, above] {Calculate ML$q$E} (MLqE);
  \draw[to, blue] (MLE) -- node[midway,above] {More robust} (MLqE);
  \draw[to] (MLE) -- (MLqE);
  \draw[to, blue] (MLE) -- node[midway, sloped, below] {Smaller variance} (XMLE);
  \draw[to] (MLE) -- node[midway, sloped, above] {Apply low-rank approximation} (XMLE);
  \draw[to, blue] (MLqE) -- node[midway, sloped, below] {Smaller variance} (XMLqE);
  \draw[to] (MLqE) -- node[midway, sloped, above] {Apply low-rank approximation} (XMLqE);
  \draw[to, blue] (XMLE) -- node[midway,above] {More robust} (XMLqE);
  \draw[to] (XMLE) -- (XMLqE);
  \draw[to, blue] (MLE) -- node[midway, sloped, above] {More robust and smaller variance} (XMLqE);
  \draw[to] (MLE) -- (XMLqE);
\end{tikzpicture}
\end{center}
\caption{\label{fig:roadmap}Roadmap among the data and four estimators.}
\end{figure}

\subsection{Entry-wise Maximum Likelihood Estimator \texorpdfstring{$\hat{\bm{P}}^{(1)}$}{$P$}}

Under the WIEM, the most natural estimator is the MLE, which happens to be the element-wise MLE $\hat{\bm{P}}^{(1)}$ in this case.
Moreover, when $\mathcal{F}$ is a one-parameter exponential family, such as Bernoulli, Poisson, or Exponential, the entry-wise MLE $\hat{\bm{P}}^{(1)}$ is the uniformly minimum-variance unbiased estimator, i.e.\ it has the smallest variance among all unbiased estimators. In addition, it has desirable asymptotic properties as the number of graphs $m$ goes to infinity.
However, in high dimensional situations such as our graph setting, the entry-wise MLE often leads to inaccurate estimates with very high variance when the sample size $m$ is small. Also, it does not exploit any graph structure. The performance will not improve as the number of vertices in each graph $n$ increases since it is an entry-wise estimator. Moreover, if the graphs are actually distributed under a WRDPG or a WSBM, then the entry-wise MLE is no longer the MLE  and the performance can be improved by considering low-rank estimators.

\subsection{Estimator \texorpdfstring{$\widetilde{\bm{P}}^{(1)}$}{$P$} Based on Adjacency Spectral Embedding of \texorpdfstring{$\hat{\bm{P}}^{(1)}$}{$P$}}

Motivated by the low-rank structure of the parameter matrix $P$ in WRDPG, we consider the estimator $\widetilde{\bm{P}}^{(1)}$ proposed by \citet{tang2016law} based on the spectral decomposition of $\hat{\bm{P}}^{(1)}$.
The construction procedure of $\widetilde{\bm{P}}^{(1)}$ consists of several steps, which will be introduced respectively in the following subsections.

\subsubsection{Rank-$d$ Approximation}

Given a dimension $d$, we consider $\widetilde{\bm{P}}^{(1)} = \mathrm{lowrank}_d(\hat{\bm{P}}^{(1)})$ as the best rank-$d$ positive semi-definite approximation of $\hat{\bm{P}}^{(1)}$. To find this best approximation, first calculate the eigen-decomposition of the symmetric matrix $\hat{\bm{P}}^{(1)} = \hat{\bm{U}} \hat{\bm{S}} \hat{\bm{U}}^{\top} + \widetilde{\bm{U}} \widetilde{\bm{S}} \widetilde{\bm{U}}^{\top}$, where $\hat{\bm{S}}$ is the diagonal matrix with the largest $d$ eigenvalues of $\hat{\bm{P}}^{(1)}$, and $\hat{\bm{U}}$ has the corresponding eigenvectors as each column. Similarly, $\tilde{\bm{S}}$ is the diagonal matrix with non-increasing entries along the diagonal corresponding to the remaining $n - d$ eigenvalues of $\hat{\bm{P}}^{(1)}$, and $\tilde{\bm{U}}$ has the columns given by the corresponding eigenvectors.
The $d$-dimensional adjacency spectral embedding (ASE) of $\hat{\bm{P}}^{(1)}$ is given by $\hat{\bm{X}}=\hat{\bm{U}} \hat{\bm{S}}^{1/2}\in \mathbb{R}^{n \times d}$; the best rank-$d$ positive semi-definite approximation  $\hat{\bm{P}}^{(1)}$ is then $\widetilde{\bm{P}}^{(1)} = \hat{\bm{X}} \hat{\bm{X}}^{\top}=\hat{\bm{U}}\hat{\bm{S}}\hat{\bm{U}}^{\top}$. The issue of how to select the emebdding dimesion $d$ is discussed later in Section~\ref{section:dim_select}.
In the RDPG setting, \citet{sussman2014consistent} proved that, provided that $d$ is chosen appropriately, each row of $\hat{\bm{X}}$ can accurately estimate the the latent position for each vertex up to an orthogonal transformation. We will extend these results for $\hat{\mathbf{X}}$ to the WRDPG setting in Section \ref{section:theory}.

Here, we restate the algorithm given in \citep{tang2016law} to give the detailed steps of computing this low-rank approximation of a general $n$-by-$n$ symmetric matrix $\bm{A}$ in Algorithm~\ref{algo:lowrank}.
\begin{algorithm}[H]
\caption{Algorithm to compute the rank-$d$ approximation of a matrix}
\label{algo:lowrank}
\begin{algorithmic}[1]
\REQUIRE Symmetric matrix $\bm{A} \in \mathbb{R}^{n \times n}$ and dimension $d\leq n$
\ENSURE $\mathrm{lowrank}_d(\bm{A})\in \mathbb{R}^{n \times n}$
\STATE Compute the algebraically largest $d$ eigenvalues of $\bm{A}$, $s_1\geq s_2\ge \dotsc \ge s_d$ and corresponding unit-norm eigenvectors $\bm{u}_1,\bm{u}_2,\dotsc,\bm{u}_d\in \mathbb{R}^n$
\STATE Set $\hat{\bm{S}}$ to the $d\times d$ diagonal matrix $\mathrm{diag}(s_1,\dotsc,s_d)$
\STATE Set $\hat{\bm{U}} = [\bm{u}_1,\dotsc,\bm{u}_d]\in \mathbb{R}^{n \times d}$
\STATE Set $\mathrm{lowrank}_d(\bm{A})$ to $\hat{\bm{U}}\hat{\bm{S}}\hat{\bm{U}}^{\top}$
\end{algorithmic}
\end{algorithm}

\subsubsection{Dimension Selection}
\label{section:dim_select}

Although Algorithm~\ref{algo:lowrank} provides a way to calculate the best rank-$d$ positive semi-definite approximation of a general symmetric matrix $\bm{A}$, it does not tell us how to select a proper dimension $d$. If we choose a relatively small dimension $d$, the estimator based on this approximation will fail to capture important information. On the other hand, when $d$ is too large, the approximation will be subject to substantial noise and also lead to a poor estimate. So a carefully selected dimension $d$ is an essential aspect of this approximation/estimation.

A general approach to selecting the dimension $d$ is to analyze the ordered eigenvalues and look for a ``gap'' or ``elbow'' in the scree-plot.
In particular, a method proposed in \citep{zhu2006automatic} finds the gaps in the scree-plot by positing a Gaussian mixture model for the ordered eigenvalues. This method provides multiple choices based on different elbows. In this paper, to avoid under-estimating the dimension, which is often much more harmful than over-estimating it, we choose the third elbow returned by the procedure of \citep{zhu2006automatic}.

Although it is always challenging to select a proper dimension, the results of our real data experiment in Section~\ref{section:real_data} demonstrate that a wide range of dimension choices lead to fairly good and comparable results. Thus any dimension selection method can generally be applied directly without excessive tuning of parameters, thereby making our low-ranked estimators useful in practice.

\begin{algorithm}[H]
\caption{Algorithm to compute $\widetilde{\bm{P}}^{(1)}$}
\label{algo:basic}
\begin{algorithmic}[1]
\REQUIRE Symmetric adjacency matrices $\bm{A}^{(1)}, \bm{A}^{(2)}, \dotsc, \bm{A}^{(m)}$, with each $\bm{A}^{(t)} \in \mathbb{R}^{n \times n}$
\ENSURE Estimate $\widetilde{\bm{P}}^{(1)} \in \mathbb{R}^{n \times n}$
\STATE Calculate the entry-wise MLE $\hat{\bm{P}}^{(1)}$
\STATE Select the dimension $d$ based on the eigenvalues of $\hat{\bm{P}}^{(1)}$; (see Section~\ref{section:dim_select})
\STATE Set $\bm{Q}$ to $\mathrm{lowrank}_d(\hat{\bm{P}}^{(1)})$; (see Algorithm~\ref{algo:lowrank})
\STATE Set $\widetilde{\bm{P}}^{(1)}$ with each entry $\widetilde{\bm{P}}^{(1)}_{ij} = \max(\bm{Q}_{ij}, 0)$
\end{algorithmic}
\end{algorithm}

With the dimension selection procedure as described above, 
the detailed description for calculating our estimator $\widetilde{\bm{P}}^{(1)}$ is then given by Algorithm~\ref{algo:basic}.

\subsection{Entry-wise Maximum L$q$-likelihood Estimator \texorpdfstring{$\hat{\bm{P}}^{(q)}$}{$P$}}

In the case of no contamination, the MLE is asymptotically efficient, i.e.\ when sample size is large enough, the MLE is at least as accurate as any other estimator. However, when the sample size is moderate, robust estimators can outperform the MLE in terms of mean squared error by winning the bias-variance tradeoff. Moreover, under contamination models, robust estimators can even outperform the MLE asymptotically since they are designed to be not unduly affected by outliers. We consider one such robust estimator, the maximum L$q$-likelihood estimator (ML$q$E) proposed by \citet{ferrari2010maximum}.

Let $X_1, \dotsc, X_m$ be sampled from $f_{\theta_0} \in \mathcal{F} = \{ f_{\theta}, \theta \in \Theta \}$, $\theta_0 \in \Theta$. Then the maximum L$q$-likelihood estimate ($q > 0$) of $\theta_0$ based on the parametric model $\mathcal{F}$ is defined as
\[
	\hat{\theta}_{\mathrm{ML}q\mathrm{E}} = \argmax_{\theta \in \Theta} \sum_{i=1}^m L_q[f_{\theta}(X_i)],
\]
where $L_q(u) = (u^{1-q} - 1)/(1- q)$.
Note that $L_q(u) \to \log(u)$ when $q \to 1$. Thus ML$q$E is a generalization of MLE.
Moreover, define
\[
	U_{\theta}(x) = \nabla_{\theta} \log f_{\theta}(x)
\]
and
\[
	U^{\star}_{\theta}(x; q) = U_{\theta}(x) f_{\theta}(x)^{1-q}.
\]
Then the ML$q$E $\hat{\theta}_{\mathrm{ML}q\mathrm{E}}$ can also be seen as a solution to the equation
\[
	\sum_{i=1}^m U^{\star}_{\theta}(X_i; q) = 0.
\]
This form interprets $\hat{\theta}_{\mathrm{ML}q\mathrm{E}}$ as a solution to a weighted likelihood equation. The weights $f_{\theta}(x)^{1-q}$ are proportional to the $(1-q)$th power of the corresponding probability. Specifically, when $0 < q < 1$, the ML$q$E puts less weight on the data points which do not fit the current distribution well. Equal weights are induced by $q=1$ and lead to the standard MLE.

Under the WIEM, we can calculate the robust entry-wise ML$q$E $\hat{\bm{P}}^{(q)}$ based on the adjacency matrices $\bm{A}^{(1)}, \dotsc, \bm{A}^{(m)}$. Note that $\hat{\bm{P}}^{(1)}$, the entry-wise MLE, is a special case of the entry-wise ML$q$E $\hat{\bm{P}}^{(q)}$ for $q = 1$.
There is also a bias-variance tradeoff in selecting the parameter $q$. \citet{qin2017robust} proposed a way to select $q$ in general. In this work, we do not focus on automatic selection of $q$.

\subsection{Estimator \texorpdfstring{$\widetilde{\bm{P}}^{(q)}$}{$P$} Based on Adjacency Spectral Embedding \texorpdfstring{$\hat{\bm{P}}^{(q)}$}{$P$}}

Intuitively, the low-rank structure of the parameter matrix $\bm{P}$ in WRDPG should be preserved approximately in the entry-wise ML$q$E $\hat{\bm{P}}^{(q)}$. Thus, in order to take advantage of such low-rank structure as well as the robustness, we apply the similar idea here as in building $\widetilde{\bm{P}}^{(1)}$, i.e.\ enforce a low-rank approximation on the entry-wise ML$q$E matrix $\hat{\bm{P}}^{(q)}$ to get $\widetilde{\bm{P}}^{(q)}$.
 The details of the construction of $\widetilde{\bm{P}}^{(q)}$ are given in Algorithm~\ref{algo:basic_q}. Algorithm~\ref{algo:basic_q} is almost identical to that of Algorithm~\ref{algo:basic} for constructing $\widetilde{\bm{P}}^{(1)}$, the main change being the use of $\hat{\bm{P}}^{(q)}$ in place of $\hat{\bm{P}}^{(1)}$. 


\begin{algorithm}[H]
\caption{Algorithm to compute $\widetilde{\bm{P}}^{(q)}$}
\label{algo:basic_q}
\begin{algorithmic}[1]
\REQUIRE Symmetric adjacency matrices $\bm{A}^{(1)}, \bm{A}^{(2)}, \dotsc, \bm{A}^{(m)}$, with each $\bm{A}^{(t)} \in \mathbb{R}^{n \times n}$
\ENSURE Estimate $\widetilde{\bm{P}}^{(q)} \in \mathbb{R}^{n \times n}$
\STATE Calculate the entry-wise ML$q$E $\hat{\bm{P}}^{(q)}$
\STATE Select the dimension $d$ based on the eigenvalues of $\hat{\bm{P}}^{(q)}$; (see Section~\ref{section:dim_select})
\STATE Set $\bm{Q}$ to $\mathrm{lowrank}_d(\hat{\bm{P}}^{(q)})$; (see Algorithm~\ref{algo:lowrank})
\STATE Set $\widetilde{\bm{P}}^{(q)}$ with each entry $\widetilde{\bm{P}}^{(q)}_{ij} = \max(\bm{Q}_{ij}, 0)$
\end{algorithmic}
\end{algorithm}


\section{Theoretical Results}
\label{section:theory}
In this section, for illustrative purposes, we present theoretical results for the case in which the contamination model introduced in Section~\ref{section:Contamination} is with respect to exponential distributions. That is $\mathcal{F} = \{ f_{\theta}(x) = \exp(-x/\theta)/{\theta}, \theta \in [0, R] \subset \mathbb{R} \}$, where $R > 0$ is a constant. These results can be extended beyond the exponential under appropriate conditions, which will be discussed in Section~\ref{section:extension}.

For clarity, we restate the model settings discussed in Section~\ref{section:Contamination}. Consider the WSBM with parameters $\bm{B}$ and $\bm{\rho}$. First we sample the block membership $\bm{\tau}$ from the categorical distribution with parameter $\bm{\rho}$ and keep it fixed for all $m$ graphs. Conditioned on this $\bm{\tau}$, the uncontaminated probability matrix $\bm{P}$  satisfies $\bm{P}_{ij} = \bm{B}_{\bm{\tau}_i, \bm{\tau}_j}$. In this section, we assume the contamination has the same block membership $\bm{\tau}$, and so the contamination matrix $\bm{C} \in \mathbb{R}^{n \times n}$ has the same block structure as $\bm{P}$.  Denote $\epsilon$ as the probability that an edge is contaminated. Then $m$ symmetric graphs $G^{(1)}, \dotsc, G^{(m)}$  are sampled such that conditioning on $\bm{\tau}$, the adjacency matrices are distributed entry-wise independently as $\bm{A}^{(t)}_{ij} \stackrel{ind}{\sim} (1-\epsilon) f_{\bm{P}_{ij}} + \epsilon f_{\bm{C}_{ij}}$ for each $1 \le t \le m$, $1 \le i < j \le n$.
Note that our theoretical results do not require the
contamination to have the same block structure and block membership $\bm{\tau}$ as the uncontaminated probability matrix; different block structure for $\mathbf{C}$ will lead to similar results -- but potentially require larger embedding dimension for $\hat{\mathbf{X}}$  -- since the rank of $(1-\epsilon) \bm{P}_{ij} + \epsilon \bm{C}_{ij}$ is still finite.

In the setting outlined above, we now analyze the performance of all four estimators based on $m$ adjacency matrices for estimating the probability matrix $\bm{P}$ in terms of the mean squared error. When comparing two estimators, we mainly focus on both asymptotic bias and asymptotic variance. Note that all the results in this section are entry-wise, which easily leads to the result for the total MSE for the entire matrix. We present the main results in this section. Additional results and proofs of stated results are given in the appendix.

\subsection{\texorpdfstring{$\hat{\bm{P}}^{(1)}$}{$P$} vs. \texorpdfstring{$\hat{\bm{P}}^{(q)}$}{$P$}}
\label{section:MLEvsMLqE}
We first compare the performance between the entry-wise MLE $\hat{\bm{P}}^{(1)}$ and the entry-wise ML$q$E $\hat{\bm{P}}^{(q)}$. Without using the graph structure, the asymptotic results for these two estimators are in terms of the number of graphs $m$, not the number of vertices $n$ within each graph.

\begin{theorem}
\label{thm:MLEvsMLqE}
For any $0 < q < 1$ and any $\bm{P}$, there exists a constant $C_0(\epsilon, q) > 0$ depending only on $\epsilon$, $q$, and $\max_{ij} \bm{P}_{ij}$ such that under the contaminated model $\bm{A}_{ij}^{(t)} \sim (1 - \epsilon) f_{\bm{P}_{ij}} + \epsilon f_{\bm{C}_{ij}}$ with $\bm{C}_{ij} > C_0(\epsilon, q)$ for all $i,j$ ML$q$E has smaller entry-wise asymptotic bias compared to MLE, i.e.\
\[
	\lim_{m \to \infty} \left| E[\hat{\bm{P}}^{(q)}_{ij}] - \bm{P}_{ij} \right| <
    \lim_{m \to \infty} \left| E[\hat{\bm{P}}^{(1)}_{ij}] - \bm{P}_{ij} \right|,
\]
for $1 \le i, j \le n$ and $i \ne j$.
Moreover, for any $\epsilon \geq 0$ and any $q$,
\[
	\mathrm{Var}(\hat{\bm{P}}^{(1)}_{ij})
    = \mathrm{Var}(\hat{\bm{P}}^{(q)}_{ij}) = O(1/m).
\]
And thus
\[
	\lim_{m \to \infty} \mathrm{Var}(\hat{\bm{P}}^{(1)}_{ij})
    = \lim_{m \to \infty} \mathrm{Var}(\hat{\bm{P}}^{(q)}_{ij}) = 0.
\]
\end{theorem}

Theorem~\ref{thm:MLEvsMLqE} shows that the entry-wise ML$q$E $\hat{\bm{P}}^{(q)}$ has smaller bias for estimating $\bm{P}$ asymptotically compared to the entry-wise MLE $\hat{\bm{P}}^{(1)}$. Although we put restrictions on the contamination matrix $\bm{C}$ in the statement of the theorem, the result still holds provided that $\epsilon (\bm{C}_{ij} - \bm{P}_{ij}) > (1 - q) \bm{P}_{ij}$ for all $i,j$. This condition can be interpreted as only requiring that the contamination of the model is large enough (either large contamination parameter matrix, or higher likelihood of encountering an outlier). From a different perspective, by putting a condition on $q$ with respect to the amount of contamination, the above condition also corresponds to only requiring that $\hat{\bm{P}}^{(q)}$ be robust enough with respect to the contamination. Thus besides the current condition for $\bm{C}$, equivalently, we can also replace it by the assumption of a large enough $\epsilon$ or a small enough $q$.

Theorem~\ref{thm:MLEvsMLqE} also indicates that both estimators have variances converging to zero as the number of graphs $m$ goes to infinity, following the asymptotic properties of minimum contrast estimates. Thus the bias term will dominate in the comparison in terms of MSE.

As a result, $\hat{\bm{P}}^{(q)}$ asymptotically reduces the bias while keeping the variance asymptotically the same as that of $\hat{\bm{P}}^{(1)}$. Thus in terms of MSE, $\hat{\bm{P}}^{(q)}$ is a better estimator than $\hat{\bm{P}}^{(1)}$ when the number of graphs $m$ is large with enough contamination.

\subsection{\texorpdfstring{$\hat{\bm{P}}^{(1)}$}{$P$} vs. \texorpdfstring{$\widetilde{\bm{P}}^{(1)}$}{$P$}}

We next analyze the effect of the ASE procedure applied to the entry-wise MLE $\hat{\bm{P}}^{(1)}$ under the contamination model, so that we can compare the performance between $\hat{\bm{P}}^{(1)}$ and $\widetilde{\bm{P}}^{(1)}$.

Before proceeding to the comparison between the two estimators, we first recall the definition of the asymptotic relative efficiency (ARE) \citep{serfling2011asymptotic}, which is an important and useful criterion to compare two estimators. Note that the original definition is for unbiased estimators. Here we adapt the definition to estimators with the same asymptotic bias.
\begin{definition}
For any parameter $\theta$ of a distribution $f$, and for estimators $\hat{\theta}^{(1)}$ and $\hat{\theta}^{(2)}$ such that $E[\hat{\theta}^{(1)}] = E[\hat{\theta}^{(2)}] = \theta^{\prime}$, $n \cdot \mathrm{Var}(\hat{\theta}^{(1)}) \to V_1(f)$ and $n \cdot \mathrm{Var}(\hat{\theta}^{(2)}) \to V_2(f)$, the ARE of $\hat{\theta}^{(2)}$ to $\hat{\theta}^{(1)}$ is given by
\[
	\mathrm{ARE}(\hat{\theta}^{(2)}, \hat{\theta}^{(1)}) = \frac{V_1(f)}{V_2(f)}.
\]
\end{definition}

By the definition above, if $\mathrm{ARE}(\hat{\theta}^{(2)}, \hat{\theta}^{(1)}) < 1$, then $\hat{\theta}^{(1)}$ has a smaller variance in its sampling distribution and thus is more efficient compared to $\hat{\theta}^{(2)}$. Combined with the fact that both estimators have the same asymptotic bias, we conclude that $\hat{\theta}^{(1)}$ is a better estimate in this case.

To compare $\hat{\bm{P}}^{(1)}$ and $\widetilde{\bm{P}}^{(1)}$, we first show that they have the same entry-wise asymptotic bias and then use the ARE criterion to compare their performance in the following theorem.

\begin{theorem}
\label{thm:MLEvsMLEASE}
Assuming that $m = O(n^b)$ for any $b > 0$, then $\widetilde{\mathbf{P}}^{(1)}$, the estimator based on the ASE of the entrywise MLE $\hat{\mathbf{P}}$, has the same entry-wise asymptotic bias as $\hat{\mathbf{P}}^{(1)}$, i.e.\
\[
	\lim_{n \to \infty} \mathrm{Bias}(\widetilde{\bm{P}}_{ij}^{(1)}) = \lim_{n \to \infty} E[\widetilde{\bm{P}}_{ij}^{(1)}] - \bm{P}_{ij} = \lim_{n \to \infty} E[\hat{\bm{P}}^{(1)}_{ij}] - \bm{P}_{ij}
    = \lim_{n \to \infty} \mathrm{Bias}(\hat{\bm{P}}_{ij}^{(1)}).
\]
In addition, for $1 \le i < j \leq n$, 
\[
	\mathrm{Var}(\widetilde{\bm{P}}_{ij}^{(1)}) = O(m^{-1} n^{-1} (\log n)^3),
	\mathrm{Var}(\hat{\bm{P}}_{ij}^{(1)}) = O(m^{-1}).
\]
Thus
\[
	\frac{\mathrm{Var}(\widetilde{\bm{P}}_{ij}^{(1)})}{\mathrm{Var}(\hat{\bm{P}}_{ij}^{(1)})}
    = O(n^{-1} (\log n)^3); \qquad
	\mathrm{ARE}(\hat{\bm{P}}_{ij}^{(1)}, \widetilde{\bm{P}}_{ij}^{(1)}) = 0.
\]
for all $1 \leq i < j \leq n$. 
\end{theorem}

Theorem~\ref{thm:MLEvsMLEASE} says that when 
$m$ is fixed or grows no faster than any polynomial with respect to $n$, the ASE procedure applied to $\hat{\bm{P}}^{(1)}$ will not affect the asymptotic bias for estimating $\bm{P}$.
Combined with the fact that the ratio of the variances of the two estimators is of order $O(n^{-1} (\log n)^3)$, we have that the ARE of $\widetilde{\mathbf{P}}_{ij}^{(1)}$ to that of $\hat{\mathbf{P}}_{ij}^{(1)}$ is $0$.
Thus $\widetilde{\bm{P}}_{ij}^{(1)}$ is a better estimate of $\mathbf{P}_{ij}$ than $\hat{\bm{P}}_{ij}^{(1)}$ for large $n$. We emphasize that the order of the ratio of the variances does not depend on $m$.


As a result, the ASE procedure applied to the entry-wise MLE $\hat{\bm{P}}^{(1)}$ helps reduce the variance while keeping the bias unchanged asymptotically, leading to a better estimate $\widetilde{\bm{P}}^{(1)}$ for $\bm{P}$ in terms of MSE.


\subsection{\texorpdfstring{$\hat{\bm{P}}^{(q)}$}{$P$} vs. \texorpdfstring{$\widetilde{\bm{P}}^{(q)}$}{$P$}}

We now proceed to analyze the effect of the ASE procedure applied to the entry-wise ML$q$E $\hat{\bm{P}}^{(q)}$ under the gross error contamination model in order to compare the performance between $\hat{\bm{P}}^{(q)}$ and $\widetilde{\bm{P}}^{(q)}$. Similarly, we first show that the two estimators have the same entry-wise asymptotic bias under appropriate conditions, and then use the ARE criterion to compare their performance in the following theorem.

\begin{theorem}
\label{thm:MLqEvsMLqEASE}
Assuming that $m = O(n^b)$ for any $b > 0$, then the estimator based on ASE of ML$q$E has the same entry-wise asymptotic bias as ML$q$E, i.e.\
\[
	\lim_{n \to \infty} \mathrm{Bias}(\widetilde{\bm{P}}_{ij}^{(q)}) = \lim_{n \to \infty} E[\widetilde{\bm{P}}_{ij}^{(q)}] - \bm{P}_{ij} = \lim_{n \to \infty} E[\hat{\bm{P}}^{(q)}_{ij}] - \bm{P}_{ij}
    = \lim_{n \to \infty} \mathrm{Bias}(\hat{\bm{P}}_{ij}^{(q)}).
\]
In addition, for $1 \le i < j \le n$ and $i \ne j$,
\[
	\mathrm{Var}(\widetilde{\bm{P}}_{ij}^{(q)}) = O(n^{-1} (\log n)^3),
	\mathrm{Var}(\hat{\bm{P}}_{ij}^{(q)}) = O(m^{-1}).
\]
Thus
\[
	\frac{\mathrm{Var}(\widetilde{\bm{P}}_{ij}^{(q)})}{\mathrm{Var}(\hat{\bm{P}}_{ij}^{(q)})}
    = O(m n^{-1} (\log n)^3).
\]
Moreover, if $m = o(n (\log n)^{-3})$, then
\[
	\mathrm{ARE}(\hat{\bm{P}}_{ij}^{(q)}, \widetilde{\bm{P}}_{ij}^{(q)}) = 0.
\]
\end{theorem}

The proof for Theorem~\ref{thm:MLqEvsMLqEASE} is almost the same as the proof for Theorem~\ref{thm:MLEvsMLEASE}. But unlike the results for comparing $\widetilde{\mathbf{P}}^{(1)}$, we are missing the term $m^{-1}$ in the variance bound for $\mathrm{Var}(\widetilde{\bm{P}}^{(q)}) = O(n^{-1} (\log n)^3)$ for arbitrary $q$ due to the structure of maximum L$q$ likelihood equation. As a result, while the ASE procedure to obtain $\widetilde{\mathbf{P}}^{(q)}$ still does not affect the asymptotic bias compared to $\hat{\mathbf{P}}^{(q)}$, the ratio of variances for $\widetilde{\mathbf{P}}^{(q)}$ and $\hat{\mathbf{P}}^{(q)}$  has an extra term $m$ compared to the ratio of variances for $\widetilde{\mathbf{P}}^{(1)}$ and $\hat{\mathbf{P}}^{(1)}$. This leads to some difference in the conclusion of Theorem~\ref{thm:MLqEvsMLqEASE} compared to that of Theorem~\ref{thm:MLEvsMLEASE}. Specifically, when $m$ is fixed, the order of the ratio of the variances of $\widetilde{\mathbf{P}}^{(q)}$ and $\hat{\mathbf{P}}^{(q)}$
is $O(n^{-1} (\log n)^3)$, which converges to $0$ as $n \to \infty$. If $m$ also increases as $n$ increases, then provided that $m$ grows on the order of $o(n (\log n)^{-3})$, the ARE of $\widetilde{\mathbf{P}}^{(q)}$ and $\hat{\mathbf{P}}^{(q)}$ still converges to $0$. 

Thus the ASE procedure applied to the entry-wise ML$q$E $\hat{\bm{P}}^{(q)}$ also helps reduce the variance while keeping the same bias asymptotically, leading to a better estimate $\widetilde{\bm{P}}^{(q)}$ for $\bm{P}$ in terms of MSE.


\subsection{\texorpdfstring{$\widetilde{\bm{P}}^{(1)}$}{$P$} vs. \texorpdfstring{$\widetilde{\bm{P}}^{(q)}$}{$P$}}

Finally, we compare the performance between $\widetilde{\bm{P}}^{(1)}$ and $\widetilde{\bm{P}}^{(q)}$ by combining the previous results.

\begin{theorem}
\label{thm:MLEASEvsMLqEASE}
For sufficiently large values of $\{\bm{C}_{ij}\}$ and any $1 \le i,j \le n$, if $m \to \infty$ at order $m = O(n^b)$ for any $b > 0$, then the estimator based on ASE of ML$q$E has smaller entry-wise asymptotic bias compared to the estimator based on ASE of MLE, i.e.\
\[
	\lim_{m, n \to \infty} \mathrm{Bias}(\widetilde{\bm{P}}_{ij}^{(1)})
    > \lim_{m, n \to \infty} \mathrm{Bias}(\widetilde{\bm{P}}_{ij}^{(q)})
\]
Moreover, if $m = O(n^b)$ for any $b > 0$, then
\[
	\lim_{n \to \infty} \mathrm{Var}(\widetilde{\bm{P}}_{ij}^{(1)})
    = \lim_{n \to \infty} \mathrm{Var}(\widetilde{\bm{P}}_{ij}^{(q)}) = 0.
\]
\end{theorem}

Theorem~\ref{thm:MLEASEvsMLqEASE} is a direct result of Theorem~\ref{thm:MLEvsMLqE}, Theorem~\ref{thm:MLEvsMLEASE}, and Theorem~\ref{thm:MLqEvsMLqEASE}.
It concludes that $\widetilde{\bm{P}}^{(q)}$ inherits the robustness from the entry-wise ML$q$E $\hat{\bm{P}}^{(q)}$ and has a smaller asymptotic bias compared to $\widetilde{\bm{P}}^{(1)}$ while both estimates have variance going to 0 as $n \to \infty$. Thus $\widetilde{\bm{P}}^{(q)}$ is a better estimator than both $\widetilde{\bm{P}}^{(1)}$ and $\hat{\bm{P}}^{(1)}$. Finally, if $m = o(n (\log n)^{-3})$, then by Theorem~\ref{thm:MLqEvsMLqEASE}, $\widetilde{\bm{P}}^{(q)}$ is also better than $\hat{\bm{P}}^{(q)}$ and hence is the best estimator among all four estimators.

\subsection{Summary of Our Four Estimators}
We summarize all four estimators and their relationships in Figure~\ref{fig:summary}.
From top to bottom in the figure, we apply ASE to construct low-rank approximations which preserve the asymptotic bias and reduce the asymptotic variance. From left to right, we underweight the outliers to construct robust estimators, so with enough contamination, whenever the number of graphs $m$ is large enough, the bias term which dominates the MSE will be improved.
Thus in Figure~\ref{fig:summary} we have quantified the qualitative roadmap introduced in Figure~\ref{fig:roadmap}.
We will evaluate these four estimators on simulated and real data experiment in Section~\ref{section:results}.

\begin{figure}
\begin{center}
\hspace*{-0.2in}
\begin{tikzpicture}[
  font=\sffamily,
  every matrix/.style={ampersand replacement=\&,column sep=2cm,row sep=2cm},
  block/.style={draw,thick,rounded corners,fill=blue!20,inner sep=.3cm},
  process/.style={draw,thick,circle,fill=blue!20},
  sink/.style={source,fill=green!20},
  datastore/.style={draw,very thick,shape=datastore,inner sep=.3cm},
  dots/.style={gray,scale=2},
  to/.style={->,>=stealth',shorten >=1pt,semithick,font=\sffamily\footnotesize},
  tofrom/.style={<->,>=stealth',shorten >=1pt,semithick,font=\sffamily\footnotesize},
  every node/.style={align=center}]

  \matrix{
    \node[block] (MLE) {$\hat{\bm{P}}^{(1)}$};
      \& \& \& \node[block] (MLqE) {$\hat{\bm{P}}^{(q)}$};\\
	\\
    \node[block] (XMLE) {$\widetilde{\bm{P}}^{(1)}$};
      \& \& \& \node[block] (XMLqE) {$\widetilde{\bm{P}}^{(q)}$}; \\
  };

  \draw[tofrom] (MLE) -- node[midway,above] {$(a)$\\$\mathrm{For}$ $\mathrm{sufficiently}$ $\mathrm{large}$ $\mathrm{values}$ $\mathrm{of}$ $\{\bm{C}_{ij}\}$, \\$\mathrm{and}$ $\mathrm{for}$ $\mathrm{any}$ $1 \le i, j \le n$, \\$\underset{m \to \infty}{\lim} \mathrm{Bias}^2(\hat{\bm{P}}_{ij}^{(1)}) > \underset{m \to \infty}{\lim} \mathrm{Bias}^2(\hat{\bm{P}}_{ij}^{(q)})$}
      node[midway,below] {$\underset{m \to \infty}{\lim} \mathrm{Var}(\hat{\bm{P}}_{ij}^{(1)}) = \underset{m \to \infty}{\lim} \mathrm{Var}(\hat{\bm{P}}_{ij}^{(q)}) = 0$} (MLqE);
  \draw[tofrom] (MLE) -- node[midway,left] {$(b)$\\$\mathrm{For}$ $\mathrm{any}$ $\mathrm{fixed}$ $m$, $\mathrm{or}$ \\ $m = O(n^b)$ $\mathrm{for}$ $b>0$, \\$\mathrm{any}$ $1 \le i, j \le n$, \\$\underset{n \to \infty}{\lim} \mathrm{Bias}(\hat{\bm{P}}_{ij}^{(1)}) =$\\$ \underset{n \to \infty}{\lim} \mathrm{Bias}(\widetilde{\bm{P}}_{ij}^{(1)})$\\$\mathrm{Var}(\widetilde{\bm{P}}_{ij}^{(1)})/\mathrm{Var}(\hat{\bm{P}}_{ij}^{(1)})$\\$ = O(n^{-1} (\log n)^3)$} (XMLE);
  \draw[tofrom] (MLqE) -- node[midway,right] {$(c)$\\$\mathrm{For}$ $\mathrm{any}$ $\mathrm{fixed}$ $m$, $\mathrm{or}$ \\ $m = O(n^b)$ $\mathrm{for}$ $b>0$, \\$\mathrm{any}$ $1 \le i, j \le n$, \\$\underset{n \to \infty}{\lim} \mathrm{Bias}(\hat{\bm{P}}_{ij}^{(q)}) = $\\$\underset{n \to \infty}{\lim} \mathrm{Bias}(\widetilde{\bm{P}}_{ij}^{(q)})$\\$\mathrm{Var}(\widetilde{\bm{P}}_{ij}^{(q)})/\mathrm{Var}(\hat{\bm{P}}_{ij}^{(q)})$\\$= O(m n^{-1} (\log n)^3)$} (XMLqE);
  \draw[tofrom] (XMLE) -- node[midway,above] {$(d)$\\$\mathrm{For}$ $\mathrm{sufficiently}$ $\mathrm{large}$ $\mathrm{values}$ $\mathrm{of}$ $\{\bm{C}_{ij}\}$, \\$\mathrm{if}$ $m \to \infty$, $m = O(n^b)$  $\mathrm{for}$ $b>0$,\\ $\mathrm{any}$ $1 \le i, j \le n$, \\$\underset{m, n \to \infty}{\lim} \mathrm{Bias}^2(\widetilde{\bm{P}}_{ij}^{(1)}) > \underset{m, n \to \infty}{\lim} \mathrm{Bias}^2(\widetilde{\bm{P}}_{ij}^{(q)})$}
      node[midway,below] {$\mathrm{If}$ $m = O(n^b)$  $\mathrm{for}$ $b>0$, $\mathrm{any}$ $1 \le i, j \le n$, \\$\underset{n \to \infty}{\lim} \mathrm{Var}(\widetilde{\bm{P}}_{ij}^{(1)}) = \underset{n \to \infty}{\lim} \mathrm{Var}(\widetilde{\bm{P}}_{ij}^{(q)}) = 0$} (XMLqE);
\end{tikzpicture}
\end{center}
\caption{\label{fig:summary}Relationships among our four estimators.}
\end{figure}
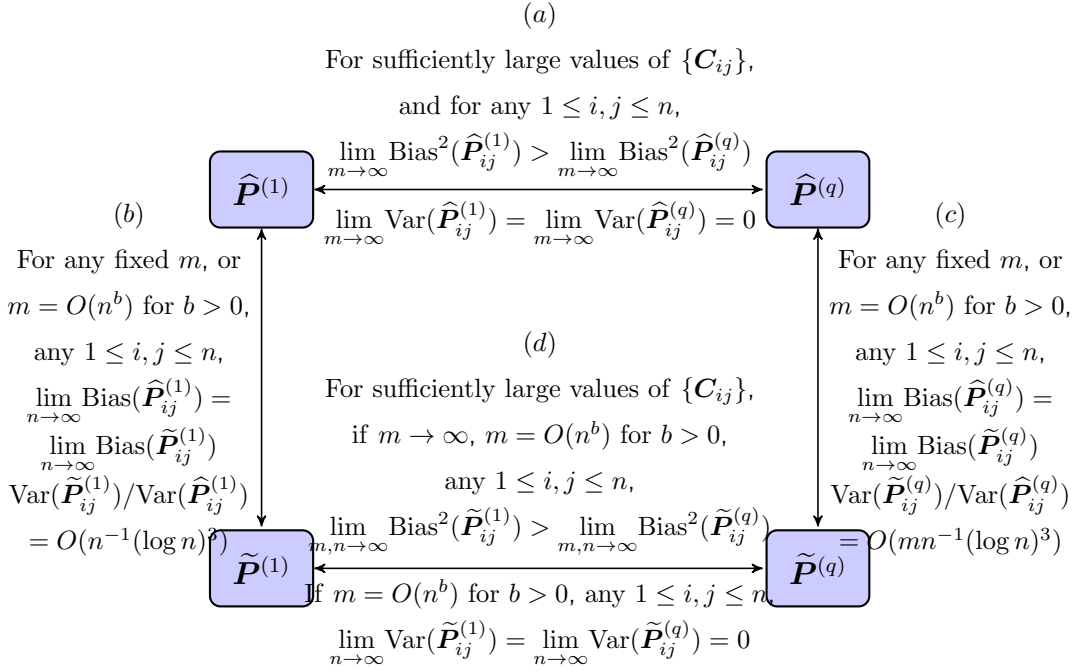

\section{Extensions}
\label{section:extension}
The results in Section~\ref{section:theory} are presented in the specific setting of the exponential distribution and the robust ML$q$E estimator. These results, however, can also be generalized to a broader class of distribution families and/or  different entry-wise robust estimators (denoted as $\hat{\bm{P}}^{(R)}$), provided that the following conditions are satisfied:
\begin{enumerate}
\item Letting $\bm{A}_{ij} \stackrel{ind}{\sim} (1-\epsilon) f_{\bm{P}_{ij}} + \epsilon f_{\bm{C}_{ij}}$, we require $f_\theta$ to satisfies $E[(\bm{A}_{ij} - E[\hat{\bm{P}}_{ij}^{(1)}])^k] \le \mathrm{const}^k \cdot k!$, where $\hat{\bm{P}}^{(1)}$ is the entry-wise MLE;
\item There exists $C_0(\bm{P}_{ij}, \epsilon) > 0$ such that under the contaminated model with $\bm{C}_{ij} > C_0(\bm{P}_{ij}, \epsilon)$,
\[
	\lim_{m \to \infty} \left| E[\hat{\bm{P}}^{(R)}_{ij}] - \bm{P}_{ij} \right| <
    \lim_{m \to \infty} \left| E[\hat{\bm{P}}^{(1)}_{ij}] - \bm{P}_{ij} \right|;
\]
\item $\hat{\bm{P}}^{(R)}_{ij} \le \mathrm{const} \cdot \hat{\bm{P}}_{ij}^{(1)}$;
\item $\mathrm{Var}(\hat{\bm{P}}^{(R)}_{ij}) = O(m^{-1})$, where $m$ is the number of graph observations.
\end{enumerate}

Condition~1 is to ensure that the observations $\mathbf{A}_{ij}$ does not deviate too far from their expectation so that Bernstein-like concentration inequalities can be applied. 
Condition~2 is similar to that used in Section \ref{section:MLEvsMLqE}; in particular, it assumes that the contamination of the model is large enough (a restriction on the distribution) and/or $\hat{\bm{P}}^{(R)}$ is sufficiently robust with respect to the contamination (a condition on the estimator).
By taking advantage of Condition 1 which controls $\hat{\bm{P}}^{(1)}$, Condition~3 allows one to derive Bernstein-like concentration inequalities for $\hat{\bm{P}}^{(R)}$.
Condition~4 ensures that the variance of $\hat{\bm{P}}^{(R)}_{ij}$ is comparable to the variance of the entry-wise MLE $\hat{\bm{P}}^{(1)}_{ij}$, which is of order $O(m^{-1})$. (Note that in the absence of Condition~4, similar but weaker results can still be derived.)

As an example to clarify the above four conditions, we sketch the argument that the results of Section~\ref{section:theory}
also holds for the
Poisson distribution when the entry-wise ML$q$E estimator is used. The Poisson distribution is a commonly used distribution for nonnegative graphs with integer weights. Lemma~\ref{lm:poisson} verifies Condition~1; intuitively, since the exponential distribution has a fatter tail compared to the Poisson, we should have the bound for the central moment of the Poisson directly from the results for the exponential distribution. Condition~2 is satisfied when we use the ML$q$E with the Poisson distribution. More specifically, since the gross error contamination is to the right, i.e.\ $C_{ij} > P_{ij}$, the weights in the ML$q$E equation will be smaller when the observed values are larger observations, in contrast to the equal weighting in the MLE equation. 
Thus under the gross error model, for sufficiently large $m$, $\hat{\bm{P}}_{ij}^{(R)}$ will be less biased than $\hat{\bm{P}}_{ij}^{(1)}$. For condition 3, $\hat{\bm{P}}^{(R)}_{ij}/\hat{\bm{P}}^{(1)}_{ij}$ is maximized when there are $m$ data points $x_1, \cdots, x_m$ with $0 \le x_1 = \cdots = x_k \le \bar{x} \le x_{k+1} = \cdots = x_m \le m \bar{x}/(m - k)$. In order to have ML$q$E larger than MLE $\bar{x}$, we need the weights of the first $m$ data points to be smaller than the weights of the remaining $m - k$ points. Thus $\exp(-\bar{x}) < \bar{x}^{x_m} \exp(-\bar{x}) / x_m!$. But then $x_m! < \bar{x}^{x_m}$. By the lower bound in Stirling's formula, we have $x_m < e \cdot \bar{x}$ when $x_m > 0$. Note that if $x_m = 0$ then MLE equals ML$q$E since all data points equal zero. Thus ML$q$E is bounded by $e \cdot \bar{x}$. As a result, $\hat{\bm{P}}_{ij} \le e \cdot \hat{\bm{P}}_{ij}^{(1)}$ and Condition~3 is satisfied. Finally, Condition~4 follows directly from the theory of minimum contrast estimators.
In summary, all theorems in Section~\ref{section:theory} hold for the Poisson distribution together with the ML$q$E. The four conditions presented in this section provide a general framework for extending the theory to more general models and robust estimators.

\section{Empirical Results}
\label{section:results}

\subsection{Simulation}
\label{section:sim}


\subsubsection{Simulation Setting}
\label{section:sim_setting}
Here we consider a 2-block WSBM with respect to the exponential distribution parameterized by
\begin{equation*}
\bm{B} = \begin{bmatrix}
4 & 2 \\
2 & 7
\end{bmatrix}
,\qquad \bm{\rho} = \begin{bmatrix}
0.5 & 0.5
\end{bmatrix}.
\end{equation*}
Let the contamination also be a 2-block WSBM with the same structure parameterized by
\begin{equation*}
\bm{B}^{\prime} = \begin{bmatrix}
9 & 6 \\
6 & 13
\end{bmatrix}
,\qquad \bm{\rho} = \begin{bmatrix}
0.5 & 0.5
\end{bmatrix}.
\end{equation*}
With these parameters specified, we sample graphs according to Section~\ref{section:Contamination}.

For ease of presentation, in the simulation we assume the true dimension $d = \mathrm{rank}(\bm{B}) = 2$ is known, and thus we eliminate the dimension selection step in Algorithm~\ref{algo:basic} and Algorithm~\ref{algo:basic_q}.

\subsubsection{Diagonal Augmentation}
\label{section:diag_aug}

Since the graphs considered in this paper have no self-loops, all the adjacency matrices $\bm{A}^{(t)}$ ($1 \le t \le m$) are hollow, i.e.\ all diagonal entries are zeros. Thus the diagonal of the parameter matrix $\bm{P}$ does not matter since all off-diagonal entries are independent of the diagonal conditioned on the off-diagonal entries of $\bm{P}$.

However, unlike the entry-wise estimators, e.g.\ $\hat{\bm{P}}^{(1)}$, the estimators which take advantage of the graph structure can benefit from the information from the diagonals. As a result, the zero diagonals of the observed graphs will lead to unnecessary biases in those estimates.

To compensate for such unnecessary biases, \citet{marchette2011vertex} suggested using the average of the non-diagonal entries of the corresponding row as the diagonal entry before embedding. Also, \citet{scheinerman2010modeling} proposed an iterative method, which gives a different approach to resolving this issue.

As suggested in \citep{tang2016law}, in this work we combine both ideas by first using Marchette's row-averaging method and then one step of Scheinerman's iterative method.

\subsubsection{Simulation Results}

To see how the performance of the four estimators vary with respect to contamination, we first run 1000 Monte Carlo replicates based on the contaminated WSBM specified in Section~\ref{section:sim_setting} with a fixed number of vertices $n = 100$ and a fixed number of graphs $m = 20$ while varying the contamination probability $\epsilon$ from $0$ to $0.4$.
Given each sample, four estimators can be computed following Algorithm~\ref{algo:basic} and Algorithm~\ref{algo:basic_q}. Since we are not focusing on how to select the parameter $q$ in the ML$q$E estimator, we shall use a fixed $q = 0.9$ unless specified otherwise. The MSE of each estimator can be estimated since the probability matrix $\bm{P}$ is known in this simulation.

The results are presented in Figure~\ref{fig:eps}. Different curves represent the simulated MSE associated with the four different estimators.
Firstly, we see that MLE $\hat{\bm{P}}^{(1)}$ is a better estimator compared to ML$q$E $\hat{\bm{P}}^{(q)}$ when there is little or no contamination (i.e.\ $\epsilon \le 0.01$ in the figure); however this estimator degrades dramatically as the contamination probability increases.
On the other hand, the ML$q$E $\hat{\bm{P}}^{(q)}$ is slightly less efficient than the MLE $\hat{\bm{P}}^{(1)}$ when the contamination probability is small,
 but is much more robust under a large contamination probability compared to the MLE.
Secondly, we see that even with a relatively small number of vertices $n = 100$, the ASE procedure which takes advantage of the low-rank structure already helps improve the performance of $\hat{\bm{P}}^{(1)}$ and lets $\widetilde{\bm{P}}^{(1)}$ win the bias-variance tradeoff. Since the ML$q$E $\hat{\bm{P}}^{(q)}$ approximately preserves the low-rank structure of the original graph, the ASE procedure also helps and makes $\widetilde{\bm{P}}^{(q)}$ a better estimate. Although both $\widetilde{\bm{P}}^{(q)}$ and $\widetilde{\bm{P}}^{(1)}$ take advantage of the low-rank structure and have reduced variances, $\widetilde{\bm{P}}^{(q)}$ constructed based on ML$q$E inherits the robustness from ML$q$E, so when the contamination probability is large enough, $\widetilde{\bm{P}}^{(q)}$ outperforms $\widetilde{\bm{P}}^{(1)}$ and degrades more slowly.

\begin{figure}
\centering
\includegraphics[width=0.8\textwidth]{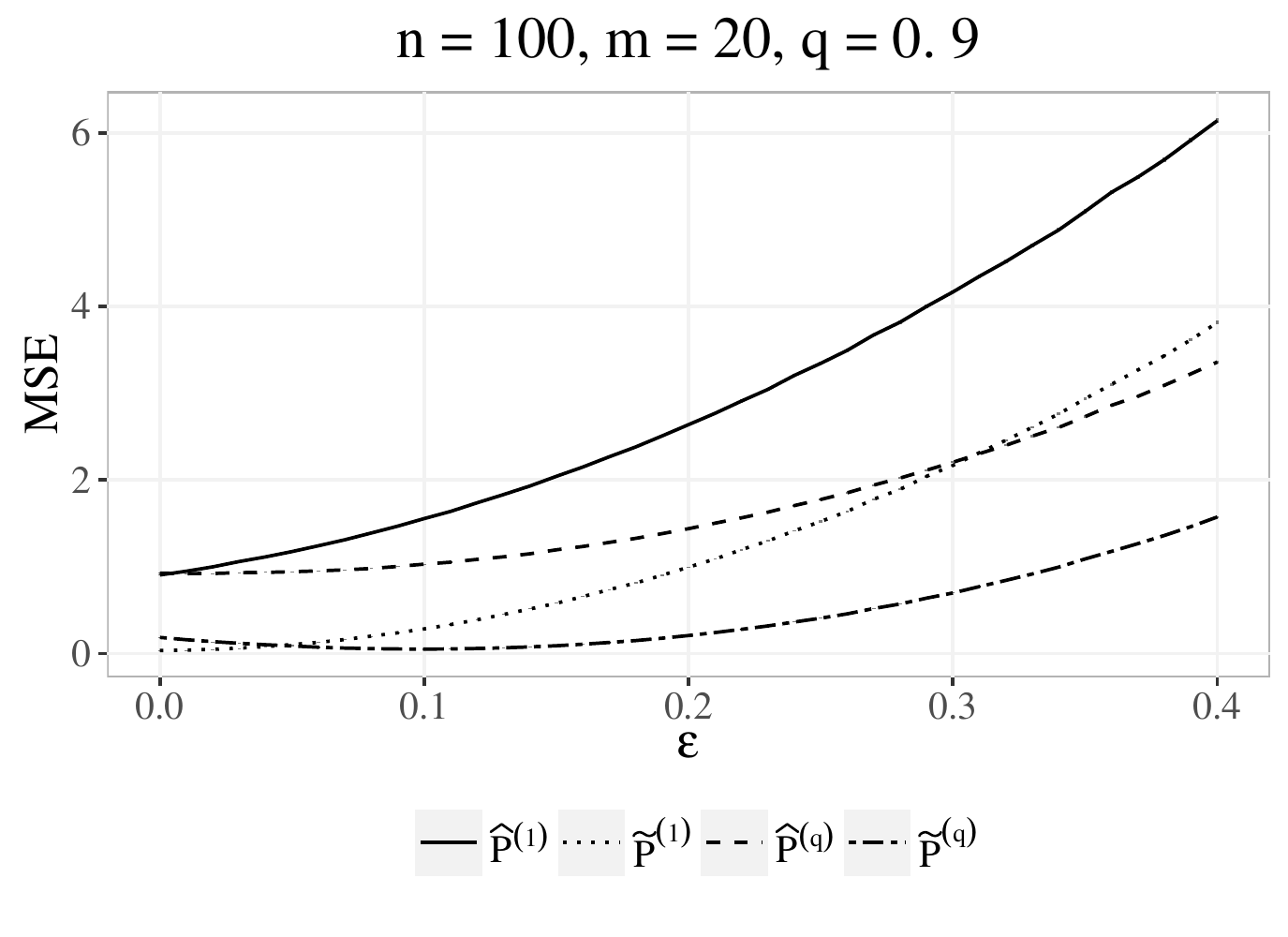}
\caption{Mean squared error in average by varying contamination ratio $\epsilon$ with fixed $n = 100$ and $m = 20$ based on 1000 Monte Carlo replicates, using $q=0.9$ when applying ML$q$E.
Different curves represent the simulated MSE associated with four different estimators.
1. MLE $\hat{\bm{P}}^{(1)}$ vs. ML$q$E $\hat{\bm{P}}^{(q)}$ (Relationship (a) in Figure~\ref{fig:summary}):
MLE outperforms by a small amount when there is no contamination (i.e.\ $\epsilon = 0$), but it degrades dramatically when contamination probability increases;
2. MLE $\hat{\bm{P}}^{(1)}$ vs. ASE $\circ$ MLE $\widetilde{\bm{P}}^{(1)}$ (Relationship (b) in Figure~\ref{fig:summary}):
ASE procedure takes the low rank structure into account and $\widetilde{\bm{P}}^{(1)}$ wins the bias-variance tradeoff;
3. ML$q$E $\hat{\bm{P}}^{(q)}$ vs. ASE $\circ$ ML$q$E $\widetilde{\bm{P}}^{(q)}$ (Relationship (c) in Figure~\ref{fig:summary}):
ML$q$E approximately preserves the low rank structure of the original graph, so ASE procedure still helps and $\widetilde{\bm{P}}^{(q)}$ wins the bias-variance tradeoff;
4. ASE $\circ$ ML$q$E $\widetilde{\bm{P}}^{(q)}$ vs. ASE $\circ$ MLE $\widetilde{\bm{P}}^{(1)}$ (Relationship (d) in Figure~\ref{fig:summary}):
When contamination probability is large enough, $\widetilde{\bm{P}}^{(q)}$ based on ML$q$E is better, since it inherits the robustness from ML$q$E.}
\label{fig:eps}
\end{figure}

Figure~\ref{fig:q} shows additional simulation results by varying the parameter $q$ in ML$q$E with fixed $n = 100$, $m = 20$ and $\epsilon = 0.1$ based on 1000 Monte Carlo replicates.
From the figure, we can see that the ASE procedure takes advantage of the graph structure and improves the performance of the corresponding estimators for a wide range of $q$. Moreover, for a wide range of $q$, the ML$q$E wins the bias-variance tradeoff and exhibits the robustness property compared to the MLE. And as $q$ goes to 1, ML$q$E goes to the MLE as expected.

\begin{figure}
\centering
\includegraphics[width=0.8\textwidth]{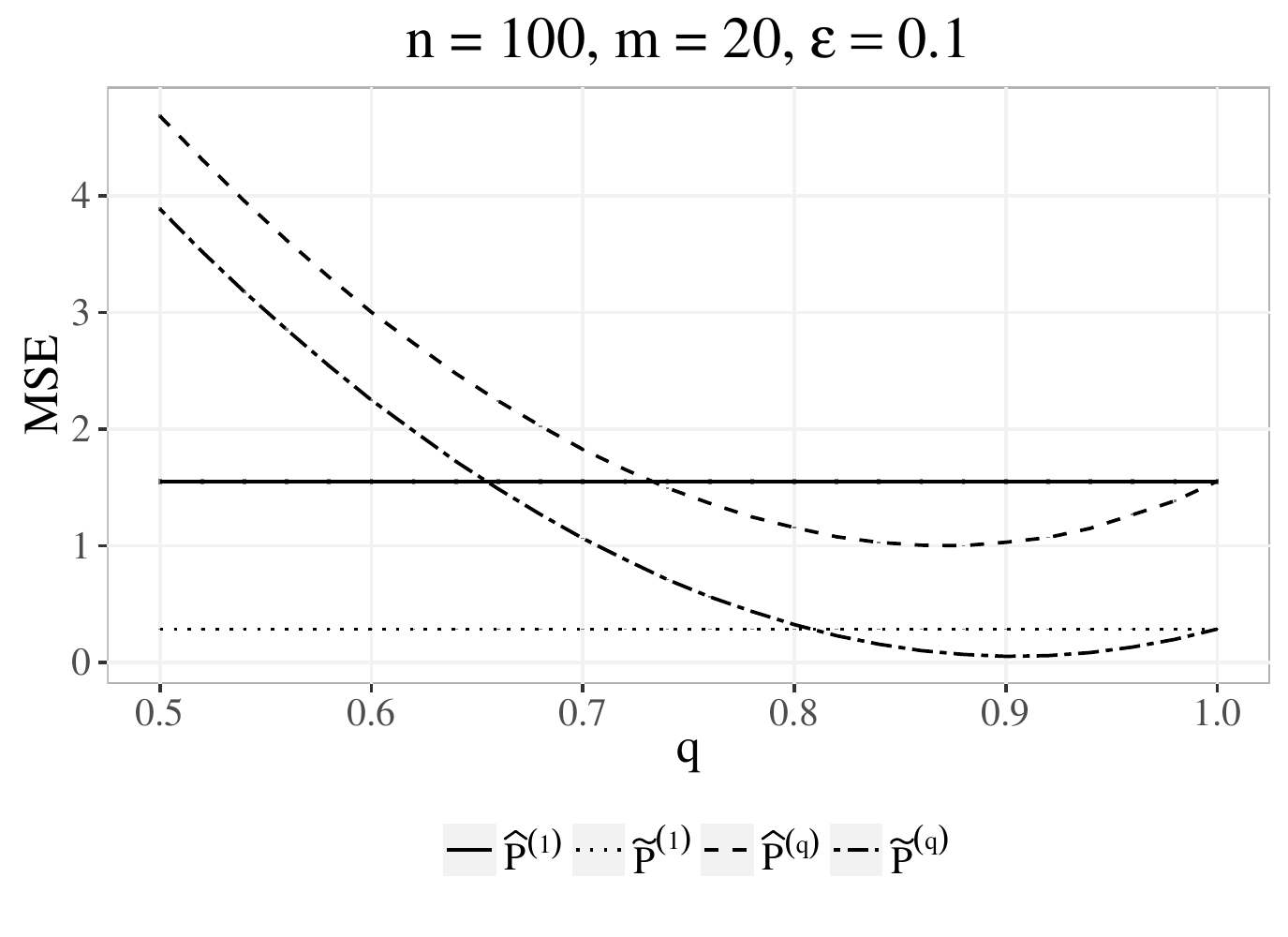}
\caption{Mean squared error in average by varying the parameter $q$ in ML$q$E with fixed $n = 100$, $m = 20$ and $\epsilon = 0.1$ based on 1000 Monte Carlo replicates. Different curves represent the simulated MSE associated with the four different estimators.
1. MLE $\hat{\bm{P}}^{(1)}$ vs. ML$q$E $\hat{\bm{P}}^{(q)}$ (Relationship (a) in Figure~\ref{fig:summary}):
Within an appropriate range of $q$, $\hat{\bm{P}}^{(q)}$ wins the bias-variance tradeoff and exhibits robustness compared to $\hat{\bm{P}}^{(1)}$. Also as $q$ goes to 1, $\hat{\bm{P}}^{(q)}$ goes to $\hat{\bm{P}}^{(1)}$ as expected;
2. MLE $\hat{\bm{P}}^{(1)}$ vs. ASE $\circ$ MLE $\widetilde{\bm{P}}^{(1)}$ (Relationship (b) in Figure~\ref{fig:summary}):
Both estimators are not affected by different choices of $q$. $\widetilde{\bm{P}}^{(1)}$ outperforms $\hat{\bm{P}}^{(1)}$ as shown in Figure~\ref{fig:eps};
3. ML$q$E $\hat{\bm{P}}^{(q)}$ vs. ASE $\circ$ ML$q$E $\widetilde{\bm{P}}^{(q)}$ (Relationship (c) in Figure~\ref{fig:summary}):
ASE procedure takes advantage of the graph structure and $\widetilde{\bm{P}}^{(q)}$ improves the performance of $\hat{\bm{P}}^{(q)}$ independent of the selection of $q$;
4. ASE $\circ$ ML$q$E $\widetilde{\bm{P}}^{(q)}$ vs. ASE $\circ$ MLE $\widetilde{\bm{P}}^{(1)}$ (Relationship (d) in Figure~\ref{fig:summary}):
Within an appropriate range of $q$, $\widetilde{\bm{P}}^{(q)}$ inherits robustness from $\hat{\bm{P}}^{(q)}$ and outperforms $\widetilde{\bm{P}}^{(1)}$.}
\label{fig:q}
\end{figure}


\subsection{Brain Graphs Experiment}
\label{section:real_data}

We now compare the four estimators on a structural connectomic dataset. The graphs in this dataset are based on diffusion tensor MR images. There are 114 different brain scans, each of which was processed to yield an undirected, weighted graph with no self-loops, using the ndmg pipeline \citep{kiar2017science, kiar2016ndmg}. There are different versions of this pipeline. In particular, the estimators are calculated based on graphs generated through version ``ndmg-v0.0.1'' (ndmg1). Note that we will consider another version ``ndmg-v0.0.33'' (ndmg2) to assess the performance of our estimators later in this section.
The vertices of the graphs represent different regions in the brain defined according to an atlas. We used the Desikan atlas with 70 vertices \citep{desikan2006automated} in this experiment. The weight of an edge between two vertices represents the number of white-matter tracts connecting the corresponding two regions of the brain.

Generally, we do not expect the graphs to perfectly follow an RDPG model, or even an IEM. Before proceeding with our analysis, we will perform some exploratory data analysis to check whether the data can reasonably be assumed to have approximate low-rank structure. Indeed, without at least some approximately low-rank structure, we will not expect the ASE procedure to improve the bias-variance tradeoff because of a potential high bias. In the left panel of Figure~\ref{fig:screehist}, we plot the eigenvalues of the mean graph of all 114 graphs (with diagonal augmentation) in decreasing algebraic order.
 The eigenvalues first decrease dramatically and then stay around 0 for a large range of dimensions. In addition, we also plot the histogram in the right panel of Figure~\ref{fig:screehist}. From the figure, we see that many eigenvalues are concentrated around zero.
This exploration suggests that the information is mostly contained in the first few dimensions. Such approximate low-rank property provides an opportunity to win the bias-variance tradeoff by applying the ASE procedure.

\begin{figure}
\centering
\includegraphics[width=.48\textwidth]{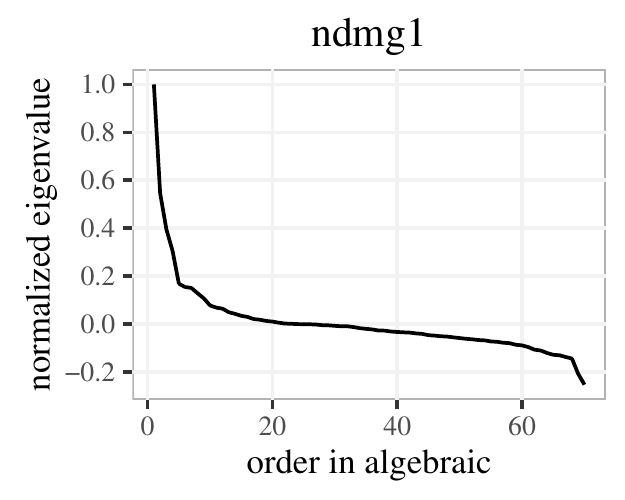}
\includegraphics[width=.48\textwidth]{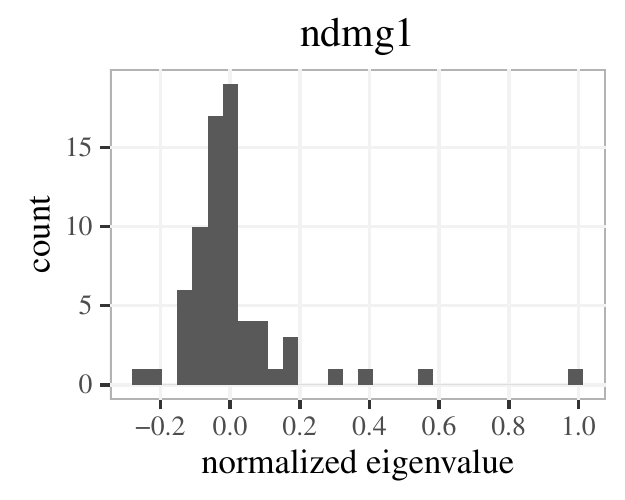}
\caption{Screeplot and the histogram of the normalized eigenvalues of the mean of 114 graphs based on ndmg1 pipeline.
The screeplot in the left panel shows the normalized eigenvalues (divided by the largest eigenvalue $230267.8$) of the mean graph of all 114 graphs with diagonal augmentation in decreasing algebraic order for the Desikan atlas. The right panel shows the histogram of the normalized eigenvalues (divided by the largest eigenvalue $230267.8$) of the mean graph of all 114 graphs with diagonal augmentation. Many eigenvalues are around zero, which lead to an approximate low-rank structure.
}
\label{fig:screehist}
\end{figure}

We now discuss an important issue with respect to this current dataset. To compare the four estimators, we need a notion of the MSE, which requires the true parameter matrix $\bm{P}$. However, unlike the simulation experiment in Section~\ref{section:sim}, $\bm{P}$ is definitely not available in practice since the 114 graphs themselves are also a sample from the population. We address this issue by finding a surrogate estimate for $\bm{P}$ and using it to calculate the MSE.
Recently, \citet{kiar2016ndmg, kiar2017science} updated the ndmg pipeline to a newer version ``ndmg-v0.0.33'' (ndmg2), which generates graphs of better quality compared to the previous version ``ndmg-v0.0.1'' (ndmg1).
So the MLE derived from the 114 graphs in ndmg2 should be a relatively more accurate estimate of the actual probability matrix $\bm{P}$ for the population. We use this as our surrogate for $\bm{P}$ when calculating the MSE. However, such a $\bm{P}$ generally has full rank, which breaks the low-rank assumptions and thus makes it harder for $\widetilde{\bm{P}}^{(1)}$ and $\widetilde{\bm{P}}^{(q)}$ to improve over $\hat{\bm{P}}^{(1)}$ and $\hat{\bm{P}}^{(q)}$. Thus any improvement arising from low-rank approximation is likely to be conservative. Moreover, it is still possible that the 114 graphs from ndmg2 contain outliers. Thus by using the MLE of the ndmg2 data as $\bm{P}$, the performance of ML$q$E-related estimators $\hat{\bm{P}}^{(q)}$ and $\widetilde{\bm{P}}^{(q)}$ are also underestimated.
In summary, our approach to constructing a workable surrogate for $\bm{P}$ relies on the availability of a better pipeline ndmg2, but is biased against both ASE-based and ML$q$E-based estimators; still, as we shall see,  ASE $\circ$ ML$q$E yields the best estimate of our surrogate $\bm{P}$.

In this experiment, we build the four estimates based on the sample of size $m$ from the ndmg1 pipeline, while using the MLE of all 114 graphs from the ndmg2 pipeline as the surrogate probability matrix $\bm{P}$. Note that diagonal augmentation procedure discussed in Section~\ref{section:diag_aug} is also applied here to compensate for the  bias introduced by the zero diagonals of the adjacency matrices.
We run 100 simulations on this dataset for different sample sizes $m = 2, 5, 10$. Specifically, in each Monte Carlo replicate, we sample $m$ graphs out of the 114 from the ndmg1 pipeline and compute the four estimates based on the $m$ sampled graphs. Once again for simplicity, we set $q$ to be 0.9 without further exploration. However, the results are qualitatively similar for many choices of $q$.
We then compare these estimates to the MLE of all 114 graphs in the ndmg2 dataset.
For the two low-rank estimators $\widetilde{\bm{P}}^{(1)}$ and $\widetilde{\bm{P}}^{(q)}$, we apply ASE into all possible dimensions, i.e.\ $d$ ranges from 1 to $n$. The MSE results are shown in Figure~\ref{fig:CCI}.

\begin{figure}
\centering
\includegraphics[width=0.8\textwidth]{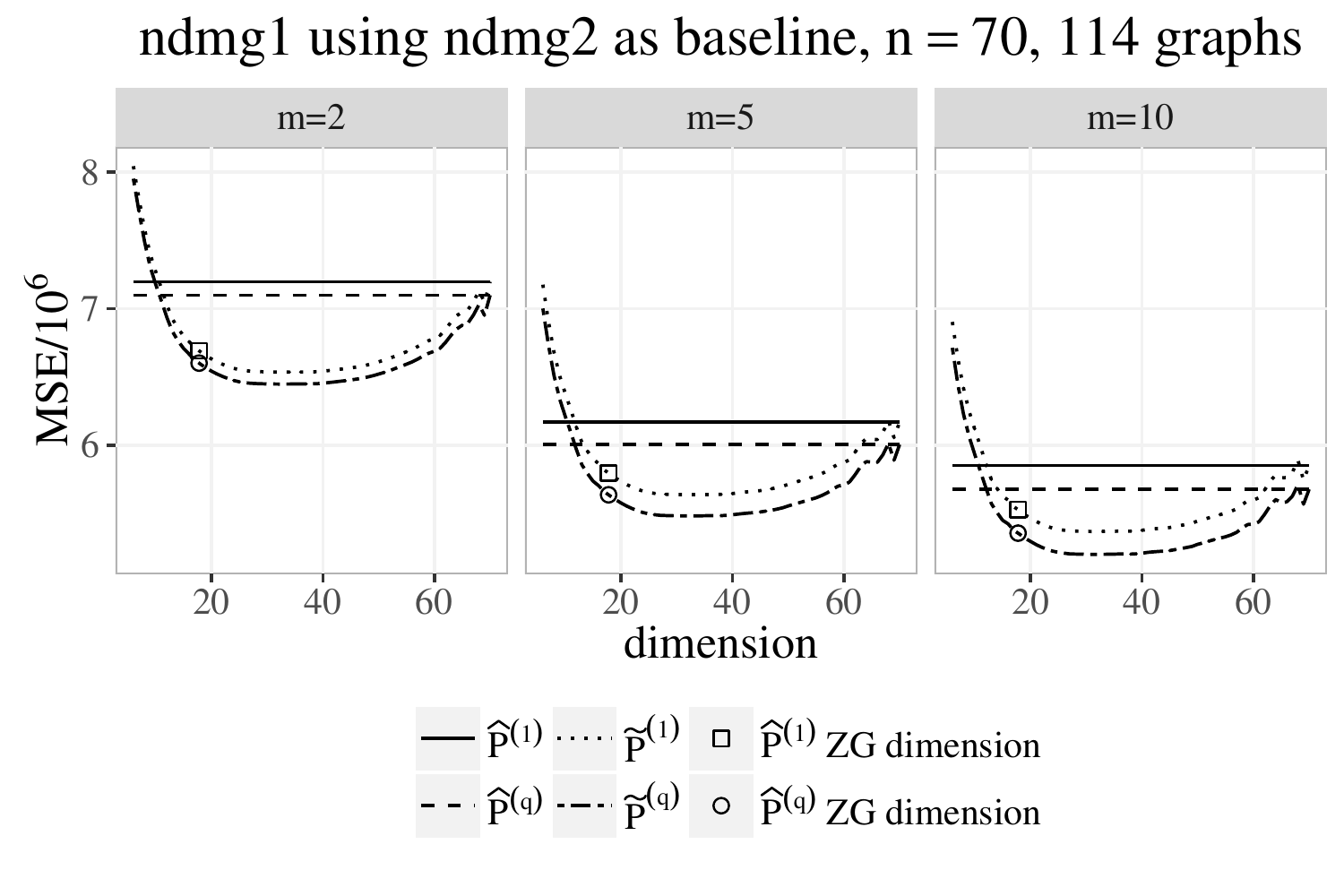}
\caption{Comparison of MSE of the four estimators using the Desikan atlas ndmg graphs at three sample sizes.
Three panels represent different sample sizes $m = 2, 5, 10$. The more samples we have, the better each estimator performs.
The horizontal-axis represents the number of embedded dimensions  while the vertical-axis characterizes the MSE of each estimator. And four estimators perform similarly according to different dimensions under different sample size $m$.
1. MLE $\hat{\bm{P}}^{(1)}$ vs. MLqE $\hat{\bm{P}}^{(q)}$:
$\hat{\bm{P}}^{(q)}$ outperforms $\hat{\bm{P}}^{(1)}$ since in practice observations are always contaminated and robust estimators are preferred;
2. MLE $\hat{\bm{P}}^{(1)}$ vs. ASE $\circ$ MLE $\widetilde{\bm{P}}^{(1)}$:
$\widetilde{\bm{P}}^{(1)}$ wins the bias-variance tradeoff when being embedded into a proper dimension;
3. MLqE $\hat{\bm{P}}^{(q)}$ vs. ASE $\circ$ ML$q$E $\widetilde{\bm{P}}^{(q)}$:
$\widetilde{\bm{P}}^{(q)}$ wins the bias-variance tradeoff when being embedded into a proper dimension;
4. ASE $\circ$ ML$q$E $\widetilde{\bm{P}}^{(q)}$ vs. ASE $\circ$ MLE $\widetilde{\bm{P}}^{(1)}$:
$\widetilde{\bm{P}}^{(q)}$ is better, since it inherits the robustness from $\hat{\bm{P}}^{(q)}$. The squares and circles represent the dimensions selected by the Zhu and Ghodsi method, which we see are reasonable choices. And more importantly, a wide range of dimensions lead to an improvement.
}
\label{fig:CCI}
\end{figure}

When $d$ is small, the ASE procedures underestimate the dimension and fail to capture important information, which leads to poor performance. In this work, we use Zhu and Ghodsi's method discussed in Section~\ref{section:dim_select} to select the dimension $d$. We denote the ZG dimension of $\hat{\bm{P}}^{(1)}$ by square and denote the ZG dimensions of $\hat{\bm{P}}^{(q)}$ by circle in the figure.
We see that the Zhu and Ghodsi's algorithm performs adequately for selecting a dimension in which to embed. More importantly, there is a wide range of dimensions for which applying ASE leads to improved performance. Although the $\bm{P}$ we are estimating is arguably not low-rank, ASE procedures still win the bias-variance tradeoff and improve performance even in this unfavorable setting.

We observed here that the robust estimator $\hat{\bm{P}}^{(q)}$ also performs relatively better than $\hat{\bm{P}}^{(1)}$, even though $\bm{P}$ (presumably) still contains outliers. This strongly indicates that there are many outliers in the original graphs from the ndmg1 pipeline, and $\widetilde{\bm{P}}^{(q)}$ successfully inherits the robustness from ML$q$E and outperforms $\widetilde{\bm{P}}^{(1)}$.

For all three sample sizes ($m = 2, 5, 10$), $\widetilde{\bm{P}}^{(q)}$ estimates the surrogate for $\bm{P}$ most accurately even though the surrogate for $\bm{P}$ is itself an edge-wise MLE, and thus is biased in favor of the other three estimators due to its high rank and non-robustness.
As such, we expect $\widetilde{\bm{P}}^{(q)}$ to provide an even better estimate for the true but unknown $\bm{P}$.

\section{Discussion}

In this work, our theoretical analysis is performed mostly in the weighted stochastic blockmodel setting. Note that the results can be extended to the weighted random dot product graph, i.e.\ our estimator does not require the block structure. Indeed, the WSBM assumption is just to ensure $\mathrm{rank}(E[\hat{\bm{P}}^{(q)}])$ has an upper bound under the contamination model that is invariant in the number of vertices. With this assumption on the rank, all the theory still holds in the WRDPG setting. In practice, graphs are not exactly low rank. However, as shown in Figure~\ref{fig:screehist} and Figure~\ref{fig:CCI}, our estimator still provides large improvement with approximate low-rank structure. Thus our method can be applied to a much more general setting instead of being restricted to WSBM.

In Section~\ref{section:theory}, we present theory based on the exponential distribution with ML$q$E for clarity. Section~\ref{section:extension} indicates that these results can be extended to other distributions and robust estimators. Note that the most important condition is Condition~1, which requires that the MLE under the corresponding edge weight and contamination distribution is concentrated so that we obtain the required matrix bounds. This generalization makes the theory more flexible and powerful.

Selecting a good distortion parameter $q$ based on real data in ML$q$E is an important but difficult task.
\citet{qin2017robust} presents a thorough, successful, and decidedly non-trivial example of such a selection methodology in the context of hypothesis testing in the one-sample univariate location problem.
While we use a fixed $q = 0.9$ in our real data experiments without presenting a formal automatic selection methodology,
Figure~\ref{fig:q} demonstrates that a poor choice of $q$ can significantly degrade performance.
We suggest that a program to develop an adaptive approach in our setting,
perhaps along the lines of \citet{qin2017robust}, is a promising avenue for widening the applicability of our estimator.

The models considered in this work assume that the vertex correspondence across graphs is known. In some applications this may not be the case.  
Our methods may still be applicable after applying graph matching algorithms such as \citep{lyzinski2016graph, lyzinski2015spectral, lyzinski2014seeded, vogelstein2015fast}.


In general, improvement in estimation performance is important not only for the estimation itself but also is important for subsequent statistical inference procedures such as clustering, vertex classification, etc.
For example, \citet{priebe2015statistical} and \citet{chen2016robust} both discuss vertex classification based on a single unweighted graph with contamination.
Additional investigation into subsequent inference tasks based on multiple contaminated weighted graphs should lead to important refinements and extensions of our robust estimation method.



\bibliography{Bib}{}
\bibliographystyle{plainnat}

\appendix

\section{Supplementary Materials: Proofs for Theory Results}

\subsection{Outline of the Proofs}

First, in Section~\ref{section:pf_MLqEvsMLE}, we prove in Lemma~\ref{lemma:ELqlEMLEproof} that when the contamination is large enough, the robust estimator $\hat{\bm{P}}^{(q)}$ has smaller asymptotic bias compared to $\hat{\bm{P}}^{(1)}$. By the results of minimum contrast estimator, we also show in Lemma~\ref{lemma:VarLqlVarMLEproof} that both estimators have variances going to zero as the number of graphs $m$ goes to infinity.

In Section~\ref{section:pf_MLEvsMLEASE1}, we analyze the properties of the ASE procedure. We first prove Theorem~\ref{thm:P1Diff}, which provides an upper bound for the spectral norm of the difference between the estimator $\hat{\bm{P}}^{(1)}$ and its expectation $\bm{H}_{ij}^{(1)} = E[\hat{\bm{P}}_{ij}^{(1)}]$. Lemma~\ref{lemma:AlmostOrthogonalL1} shows that $\bm{U}^{\top} \hat{\bm{U}}$ can be approximated by an orthogonal matrix $\bm{W}^{*} = \bm{W}_1 \bm{W}_2^{\top}$, where $\bm{U}$ and $\hat{\bm{U}}$ are the eigenspaces with respect to the largest $d$ eigenvalues of $\bm{H}_{ij}^{(1)}$ and $\hat{\bm{P}}^{(1)}$ respectively. More conveniently, Lemma~\ref{lemma:exchangeL1} indicates that we can change the order of $\bm{W}^*$ in the matrix multiplications accordingly without affecting the result much. With these tool results, in Lemma~\ref{lemma:XhatDiffXWexpressionL1} we give an upper bound of $\|\hat{\bm{Z}} - \bm{Z} \bm{W}\|_F$, which controls the error of the $\hat{\bm{Z}}$ for estimating the true latent positions $\bm{Z}$ up to orthogonal transformation.
Using Lemma~\ref{lemma:XhatDiffXWexpressionL1}, we can then give a bound for the $2 \to \infty$-norm of $\hat{\bm{Z}} - \bm{Z} \bm{W}$, i.e.\ we bound $\max_i \| \hat{\bm{Z}}_i - \bm{W} \bm{Z}_i \|_2$ in Theorem~\ref{thm:XhatDiffXWL1}.

In Section~\ref{section:pf_MLEvsMLEASE2}, we give a bound of the estimation error $\left|  \hat{\bm{Z}}_i^{\top} \hat{\bm{Z}}_j - \bm{Z}_i^{\top} \bm{Z}_j \right|$ in Lemma~\ref{lemma:1stMomentPhatDiffL1} based on the results in Section~\ref{section:pf_MLEvsMLEASE1}. In order to bound the variance of our estimator $\widetilde{\bm{P}}^{(1)}$, all results in this section will be based on a truncated version of $\widetilde{\bm{P}}^{(1)}$ defined in Definition~\ref{def:truncationMLE}. This is purely for technical reasons and will not affect the estimation procedure in practice, which is discussed in details in Remark~\ref{remark:truncation}. We then bound the expectation (Lemma~\ref{lm:L1Consistentproof}) and variance (Theorem~\ref{thm:VarASEL1proof}) of $\widetilde{\bm{P}}^{(1)}$ by carefully choosing a truncation point $a$ and applying the above truncation argument. As a direct result, we obtain the bound for the relative efficiency between $\hat{\bm{P}}_{ij}^{(1)}$ and $\widetilde{\bm{P}}_{ij}^{(1)}$ in Theorem~\ref{thm:AREL1proof}.

In Section~\ref{section:pf_MLqEASEvsMLqE}, we compare the performance between $\widetilde{\bm{P}}^{(q)}$ and $\hat{\bm{P}}^{(q)}$. The results in this section are proved in a similar manner to those in Section~\ref{section:pf_MLEvsMLEASE1} and Section~\ref{section:pf_MLEvsMLEASE2}. However, since the ML$q$E estimator for a mixture distribution model does not have a closed form expression, we explore a relationship between MLE and ML$q$E to bound $\widetilde{\bm{P}}^{(q)}$ and $\hat{\bm{P}}^{(q)}$; this technique could be of independent interest. Finally, in Section~\ref{section:MLqEASEvsMLEASE}, we compare the performance between $\widetilde{\bm{P}}^{(q)}$ and $\widetilde{\bm{P}}^{(1)}$.

In Section~\ref{section:pf_other}, we provide proofs for all supplementary results mentioned in the manuscript.

Before presenting the proofs, we first define the following notion of ``with high probability'' that is used throughout this appendix.
\begin{definition}
Let $(E_n)$ for $n \geq 1$ be a sequence of events.
We say that the events hold with high probability if, for any constant $c > 0$ there exits a constant $n_0(c)$ such that for all $n \geq n_0$, the event $E_n$ holds with probability greater than $1 - n^{-c}$.
\end{definition}

\subsection{\texorpdfstring{$\hat{\bm{P}}^{(q)}$}{$P$} vs. \texorpdfstring{$\hat{\bm{P}}^{(1)}$}{$P$}}
\label{section:pf_MLqEvsMLE}

\begin{lemma}
\label{lemma:LqlMLE}
Let $X_1, \cdots, X_m \stackrel{iid}{\sim} \mathrm{Exp}(p)$ with $m \ge 2$ and $E[X_1] = p$. Then with probability $1$,
\begin{itemize}
\item There exists at least one solution to the ML$q$ equation;
\item All the solutions to the ML$q$ equation are less than the MLE.
\end{itemize}
Thus the ML$q$E $\hat{p}^{(q)}$, the root closest to the MLE, is well defined.
\end{lemma}
\begin{proof}
Let $x_1, \cdots, x_m$ be the observed values of $X_1, X_2, \dots, X_m$. Then with probability $1$, the $x_i$ are unique and $x_{(1)} = \min_{i} x_i > 0$. The MLE is
\[
	\hat{p}^{(1)}(x) = \bar{x}.
\]

Let $g(\theta, x) = \sum_{i=1}^m e^{-\frac{(1-q)x_i}{\theta}}(x_i - \theta)$. Then the ML$q$ equation is $g(\theta, x) = 0$. Now let $l$ be the smallest index such that $x_{(1)} \le \cdots \le x_{(l)} \le \bar{x} \le x_{(l+1)} \le \cdots$. Define $s_i = \bar{x} - x_{(i)}$ for $1 \le i \le l$, and $t_{i} = x_{(l+i)} - \bar{x}$ for $1 \le i \le m - l$. Note that $\sum_{i=1}^l s_i = \sum_{i=1}^{m-l} t_i$. Then for any $\theta \ge \bar{x}$, we have
\begin{align*}
g(\theta, x) & = \sum_{i=1}^m e^{-\frac{(1-q)x_{(i)}}{\theta}}(x_{(i)} - \theta)
= \sum_{i=1}^m e^{-\frac{(1-q)x_{(i)}}{\theta}}(x_{(i)} - \bar{x} + \bar{x} - \theta) \\
& = - \sum_{i=1}^l e^{-\frac{(1-q)x_{(i)}}{\theta}}s_i
+ \sum_{i=1}^{m-l} e^{-\frac{(1-q)x_{(i+l)}}{\theta}}t_i
+ \sum_{i=1}^m e^{-\frac{(1-q)x_{(i)}}{\theta}}(\bar{x} - \theta)\\
& \le - \sum_{i=1}^l e^{-\frac{(1-q)x_{(i)}}{\theta}}s_i
+ \sum_{i=1}^{m-l} e^{-\frac{(1-q)x_{(i+l)}}{\theta}}t_i \\
& \le - e^{-\frac{(1-q)x_{(l+1)}}{\theta}} \sum_{i=1}^l s_i
+ \sum_{i=1}^{m-l} e^{-\frac{(1-q)x_{(i+l)}}{\theta}}t_i \\
& \le - e^{-\frac{(1-q)x_{(l+1)}}{\theta}} \sum_{i=1}^{m-l} t_i
+ \sum_{i=1}^{m-l} e^{-\frac{(1-q)x_{(i+l)}}{\theta}}t_i \\
& \le - \sum_{i=1}^{m-l} e^{-\frac{(1-q)x_{(i+l)}}{\theta}}t_i
+ \sum_{i=1}^{m-l} e^{-\frac{(1-q)x_{(i+l)}}{\theta}}t_i\\
& = 0,
\end{align*}
and equality holds if and only if all $x_i$'s are the same, which occurs with probability $0$. Thus with probablity $1$, $g(\theta, x) < 0$ for all $\theta \ge \bar{x}$.

Denote any solution to the ML$q$E equation as $\hat{p}^{(q)}(x)$; we then have that
\begin{itemize}
\item $g(\hat{p}^{(q)}(x), x) = 0$;
\item $\lim_{\theta \rightarrow 0^+}g(\theta, x) = 0$;
\item $g(\theta, x) > 0$ when $\theta < x_{(1)}$;
\end{itemize}

Thus there exists at least one solution to the ML$q$E equation. And since all solutions to the ML$q$E equation are in the interval $(x_{(1)},\bar{x})$, we have $\hat{p}^{(q)}(x_1, \dots, x_m) \leq \hat{p}^{(1)}(x_1, \dots, x_m)$.
\end{proof}

\begin{lemma}
\label{lemma:PopulationLqExist}
Given observed data points $x_1, \dots, x_m$, recall that the ML$q$E equation under the exponential distribution based on $m$ data points $x_1, \dots, x_m$ is $\sum_{i=1}^m \exp(-(1-q)x_i/{\theta})(x_i - \theta) = 0$ while the MLE equation under the exponential distribution based on the same data is $\sum_{i=1}^m (x_i - \theta) = 0$. Now
let $X_1, \cdots, X_m \stackrel{iid}{\sim} (1-\epsilon) \mathrm{Exp}(p) + \epsilon \mathrm{Exp}(c)$ be $m$ data points sampled from a mixture of two exponential distribution. Denoting this mixture as $F$, there exists exactly one real solution $\theta(F)$ of $E_F[\exp(-(1-q)X/{\theta(F)})(X - \theta(F))] = 0$. Note that $E_F[\exp(-(1-q)X/{\theta(F)})(X - \theta(F))] = 0$ is the population version of ML$q$E equation under the exponential distribution.
Moreover, the ML$q$E solution is less than the MLE solution under the exponential distribution, i.e.\
$\theta(F) < E_F[\bar{X}] = (1-\epsilon) p + \epsilon c$.
\end{lemma}
\begin{proof}
For the MLE, i.e.\ $\bar{X}$, we have $E[\bar{X}] = (1-\epsilon) p + \epsilon c$.
According to Equation (3.2) in \citep{ferrari2010maximum}, $\theta(F)$ satisfies
\[
\frac{\epsilon c}{(c(1-q) + \theta)^2} - \frac{\epsilon}{c(1-q) + \theta}
+\frac{(1-\epsilon) p}{(p(1-q) + \theta)^2} - \frac{(1-\epsilon)}{p(1-q) + \theta}
= 0,
\]
i.e.\
\[
\frac{\epsilon (\theta - c q)}{(c(1-q) + \theta)^2} =
\frac{(1-\epsilon) (p q - \theta)}{(p(1-q) + \theta)^2}.
\]
Define $h(\theta) = (c(1-q) + \theta)^2 (1-\epsilon) (p q - \theta) - (p(1-q) + \theta)^2 \epsilon (\theta - c q)$.
Then $\lim_{\theta \to \infty}h(\theta) = -\infty$, $h(0) > 0$, and $h(c q) < 0$.
Consider $q$ as the variable and solve the equation $h(E[\bar{X}]) = 0$, we have three roots and one of them is $q = 1$ obviously.
The other two roots are
\[
\frac{(p + c)\left( (p - c)^2 \epsilon(1-\epsilon) + 2 p c \right)}{2pc(p \epsilon + c(1-\epsilon))}
\pm \sqrt{\frac{\epsilon(1-\epsilon)(c-p)^2\left(\epsilon(1-\epsilon)(c-p)^4 - 4p^2c^2\right)}{4 p^2 c^2 (p\epsilon + c(1-\epsilon))^2}}.
\]
To prove the roots are greater or equal to 1, we need to show
\[
\frac{(p + c)\left( (p-c)^2 \epsilon(1-\epsilon) + 2pc \right)}{2pc(p \epsilon + c(1-\epsilon))}
- \sqrt{\frac{\epsilon(1-\epsilon)(c-p)^2\left(\epsilon(1-\epsilon)(c-p)^4 - 4p^2 c^2\right)}{4 p^2 c^2 (p\epsilon + c(1-\epsilon))^2}} > 1.
\]
For the first part,
\[
\frac{(p + c)\left( (p - c)^2 \epsilon(1-\epsilon) + 2p c \right)}{2pc(p \epsilon + c(1-\epsilon))}\\
> 1 + \frac{(p-c)^2 \epsilon (1-\epsilon) (p+c)}{2pc(p \epsilon + c(1-\epsilon))}.
\]
To prove the roots are greater or equal to 1, we just need to show
\[
(p-c)^4 \epsilon^2 (1-\epsilon)^2 (p+c)^2 \ge \epsilon^2(1-\epsilon)^2(c-p)^6.
\]
Then it is sufficient to show that
\[
(p+c)^2 \ge (c-p)^2,
\]
which is true.
Combined with the fact that when $q = 0$, $h(E[\bar{X}]) < 0$, we have for any $0 < q < 1$, $h(E[\bar{X}]) < 0$.

The equation $h(\theta) = 0$ is a cubic polynomial, so it has at most three real roots. In addition, by calculating we know there is only one real root, while the other two are complex roots. Combined with the fact that $h(p q) > 0$, we have for any $0 < q < 1$, the only real root of the population version of ML$q$ equation is less than $E[\bar{X}] = (1-\epsilon)p + \epsilon c$.
\end{proof}

\begin{lemma}[Theorem~\ref{thm:MLEvsMLqE}]
\label{lemma:ELqlEMLEproof}
For any $0 < q < 1$ and any $\bm{P}$, there exists a constant $C_0(\epsilon, q) > 0$ depending only on $\epsilon$, $q$, and $\max_{ij} \bm{P}_{ij}$ such that under the contaminated model $\bm{A}_{ij}^{(t)} \sim (1 - \epsilon) f_{\bm{P}_{ij}} + \epsilon f_{\bm{C}_{ij}}$ with $\bm{C}_{ij} > C_0(\epsilon, q)$ for all $i,j$ ML$q$E has smaller entry-wise asymptotic bias compared to MLE, i.e.\
\[
	\lim_{m \to \infty} \left| E[\hat{\bm{P}}^{(q)}_{ij}] - \bm{P}_{ij} \right| <
    \lim_{m \to \infty} \left| E[\hat{\bm{P}}^{(1)}_{ij}] - \bm{P}_{ij} \right|,
\]
for $1 \le i, j \le n$ and $i \ne j$.
\end{lemma}
\begin{proof}
For the MLE $\hat{\bm{P}}^{(1)}_{ij} = \bar{\bm{A}}_{ij}$,
\[
	E[\hat{\bm{P}}^{(1)}_{ij}] = E[\bar{\bm{A}}_{ij}]
    = \frac{1}{m} \sum_{t = 1}^m E[\bm{A}_{ij}^{(t)}]
    = E[\bm{A}_{ij}^{(1)}]
    = (1-\epsilon) \bm{P}_{ij} + \epsilon \bm{C}_{ij}.
\]
As shown in Lemma \ref{lemma:PopulationLqExist}, $\theta(F)$ satisfies
\[
\frac{\epsilon (\theta(F) - \bm{C}_{ij}q)}{(\bm{C}_{ij}(1-q) + \theta(F))^2} =
\frac{(1-\epsilon) (\bm{P}_{ij} q - \theta(F))}{(\bm{P}_{ij}(1-q) + \theta(F))^2}.
\]
Thus $\theta(F) - \bm{C}_{ij} q$ and $\theta(F) - \bm{P}_{ij} q$ should have different signs. Combined with $\bm{C}_{ij} > \bm{P}_{ij}$, we have
\[
q \bm{P}_{ij} < \theta(F).
\]
To have a smaller asymptotic bias in absolute value, combined with Lemma \ref{lemma:ELqConverge}, we need
\[
|\theta(F) - \bm{P}_{ij}| < \epsilon (\bm{C}_{ij} - \bm{P}_{ij}).
\]
Based on Lemma \ref{lemma:LqlMLE}, we need
\[
q \bm{P}_{ij} > \bm{P}_{ij} - \epsilon(\bm{C}_{ij} - \bm{P}_{ij}),
\]
i.e.\
\[
\bm{C}_{ij} > \bm{P}_{ij} + \frac{(1-q) \bm{P}_{ij}}{\epsilon} = C_0(\bm{P}_{ij}, \epsilon, q).
\]
\end{proof}

\begin{lemma}
\label{lemma:minimumcontrast}
Assume similar setting as Lemma~\ref{lemma:PopulationLqExist}, i.e.\
let $m$ data points sample from the contamination model $X, X_1, \cdots, X_m \stackrel{iid}{\sim} (1-\epsilon) \mathrm{Exp}(p) + \epsilon \mathrm{Exp}(c)$. The ML$q$E under the exponential distribution $\mathrm{Exp}(p)$ is a minimum contrast estimator.
\end{lemma}
\begin{proof}
Consider the contaminated distribution $F(x) = (1-\epsilon) f(x; p) + \epsilon f(x; c)$, where $f(x)$ represents the pdf of exponential distribution. By Lemma~\ref{lemma:PopulationLqExist}, we know there is a one-to-one correspondence between the uncontaminated parameter $p$ and the only real solution $\theta(F)$ of the population version of ML$q$ equation,
i.e.\ $E_F[exp(-(1-q)X/{\theta(F)})(X - \theta(F))] = 0$. Let $r(\theta(F)) = p$.
Then we can define $\rho(x; \theta) = f(x; r(\theta))^{1-q}/{1 - q}$, where $q \in (0, 1)$ is a constant.
By reparameterizing $\rho(x; \theta)$ to $\widetilde{\rho}(x; r)$ such that $\widetilde{\rho}(x; r(\theta)) = \rho(x; \theta)$, we can use the proof of Lemma~\ref{lemma:PopulationLqExist} directly to prove that $D(\theta_0, \theta) = E_{\theta_0}[\rho(X, \theta)]$ is uniquely minimized at $\theta_0$. Thus the ML$q$E is a minimum contrast estimator.
\end{proof}

\begin{lemma}
\label{lemma:UniformConvergence}
Uniform convergence of the MLq equation, i.e.\
\[
	\sup_{\theta \in [0, R]} \left| \frac{1}{m} \sum_{i=1}^m e^{-\frac{(1-q) X_i}{\theta}}(X_i - \theta) - E_F[e^{-\frac{(1-q) X}{\theta}}(X - \theta)] \right| \stackrel{a.s.}{\to} 0.
\]
\end{lemma}
\begin{proof}
Define $g(x,\theta) = \exp(-(1-q) x/{\theta})(x - \theta)$ and $d(x) = \exp(-(1-q)x/{R})(x + R)$. Then $E_F[d(X)] < \infty$ and $g(x,\theta) \le d(x)$ for all $\theta \in [0, R]$.
Combined with the fact that $[0, R]$ is compact and the function $g(x,\theta)$ is continuous at each $\theta$ for all $x > 0$ and measurable function of $x$ at each $\theta$, we have the uniform convergence by Lemma 2.4 in \citep{newey1994large}.
\end{proof}

\begin{lemma}
\label{lemma:ELqConverge}
$\hat{\bm{P}}_{ij}^{(q)} \stackrel{P}{\to} \theta(F_{ij})$ as $m \to \infty$, where $F_{ij}$ is the contaminated distribution $(1-\epsilon) \mathrm{Exp}(\bm{P}_{ij}) + \epsilon \mathrm{Exp}(\bm{C}_{ij})$, and $ \theta(F_{ij})$ is defined in Lemma~\ref{lemma:PopulationLqExist}.
\end{lemma}
\begin{proof}
By the proof of Lemma~\ref{lemma:PopulationLqExist}, we have
\[
	\inf\{D(\theta_0, \theta): |\theta - \theta_0| \ge \epsilon \} > D(\theta_0, \theta_0)
\]
for every $\epsilon > 0$. Combined with Lemma~\ref{lemma:UniformConvergence}, we know the ML$q$ is consistent based on Theorem 5.2.3 in \citep{bickel2007mathematical}.
\end{proof}

\begin{lemma}[Theorem~\ref{thm:MLEvsMLqE}]
\label{lemma:VarLqlVarMLEproof}
For $1 \le i, j \le n$,
\[
	\mathrm{Var}(\hat{\bm{P}}^{(1)}_{ij})
    = \mathrm{Var}(\hat{\bm{P}}^{(q)}_{ij}) = O(1/m).
\]
And thus
\[
	\lim_{m \to \infty} \mathrm{Var}(\hat{\bm{P}}^{(1)}_{ij})
    = \lim_{m \to \infty} \mathrm{Var}(\hat{\bm{P}}^{(q)}_{ij}) = 0.
\]
\end{lemma}
\begin{proof}
Both MLE and ML$q$E are minimum constrast estimators. By consistency (shown in Lemma~\ref{lemma:ELqConverge}) and other regularity conditions, we know the variances are both of order $1/m$ based on Theorem 5.4.2 in \citep{bickel2007mathematical}.
\end{proof}

\subsection{ASE Procedure of \texorpdfstring{$\hat{\bm{P}}^{(1)}$}{$P$}}
\label{section:pf_MLEvsMLEASE1}

\begin{theorem}
\label{thm:P1Diff}
Let $\bm{P}$ and $\bm{C}$ be two $n$-by-$n$ symmetric matrices satisfying element-wise conditions $0 < \bm{P}_{ij} \le \bm{C}_{ij} \le R$ for some constant $R > 0$. For $0 < \epsilon < 1$, we define $m$ symmetric and hollow matrices as
\[
	\bm{A}^{(t)} \stackrel{iid}{\sim} (1-\epsilon) \mathrm{Exp}(\bm{P}) + \epsilon \mathrm{Exp}(\bm{C}),
\]
for $1 \le t \le m$.
Let $\hat{\bm{P}}^{(1)}$ be the element-wise MLE based on exponential distribution with $m$ observations.
Define $\bm{H}_{ij}^{(1)} = E[\hat{\bm{P}}_{ij}^{(1)}] = (1-\epsilon) \bm{P}_{ij} + \epsilon \bm{C}_{ij}$,
then for any constant $c > 0$, there exists another constant $n_0(c)$, independent of $n$, $\bm{P}$, $\bm{C}$ and $\epsilon$, such that if $n > n_0$, then for all $\eta$ satisfying $n^{-c} \le \eta \le 1/2$,
\[
	P \left( \| \hat{\bm{P}}^{(1)} - \bm{H}^{(1)} \|_2 \le 4 R \sqrt{n \ln(n/\eta)/m}\right) \ge 1 - \eta.
\]
\end{theorem}
\textbf{Remark:} This is an extended version of Theorem 3.1 in \citep{oliveira2009concentration}.

Define $\bm{H}^{(1)} = E[\hat{\bm{P}}^{(1)}] = (1-\epsilon) \bm{P} + \epsilon \bm{C}$, where $\bm{P} = \bm{X} \bm{X}^{\top}$, $\bm{X} \in \mathbb{R}^{n \times d}$, $\bm{C} = \bm{Y} \bm{Y}^{\top}$, $\bm{Y} \in \mathbb{R}^{n\times d'}$.
Let $d^{(1)} = \mathrm{rank}(\bm{H}^{(1)})$ be the dimension in which we are going to embed $\hat{\bm{P}}^{(1)}$. Then we can define $\bm{H}^{(1)} = \bm{Z} \bm{Z}^{\top}$ where $\bm{Z} \in \mathbb{R}^{n \times d^{(1)}}$.
Since $\bm{H}^{(1)} = [\sqrt{1-\epsilon} \bm{X}, \sqrt{\epsilon} \bm{Y}] [\sqrt{1-\epsilon} \bm{X}, \sqrt{\epsilon} \bm{Y}]^{\top}$, we have $d^{(1)} \le d+d'$.

For simplicity, from now on, we will use $\hat{\bm{P}}$ to represent $\hat{\bm{P}}^{(1)}$, use $\bm{H}$ to represent $\bm{H}^{(1)}$ and use $k$ to represent the dimension $d^{(1)}$ we are going to embed. Assume $\bm{H} = \bm{U} \bm{S} \bm{U}^{\top} = \bm{Z} \bm{Z}^{\top}$, where $\bm{Z} = [\bm{Z}_1, \cdots, \bm{Z}_n]^{\top}$ is a $n$-by-$k$ matrix. Then our estimate for $\bm{Z}$ up to rotation is $\hat{\bm{Z}} = \hat{\bm{U}} \hat{\bm{S}}^{1/2}$, where $\hat{\bm{U}} \hat{\bm{S}} \hat{\bm{U}}^{\top}$ is the rank-$k$ spectral decomposition of $|\hat{\bm{P}}| = (\hat{\bm{P}}^{\top} \hat{\bm{P}})^{1/2}$.

Furthermore, we assume that the second moment matrix $E[\bm{Z}_1 \bm{Z}_1^{\top}]$ is rank $k$ and has distinct eigenvalues $\lambda_i(E[\bm{Z}_1 \bm{Z}_1^{\top}])$. In particular, we assume that there exists $\delta > 0$ such that
\[
	\delta < \lambda_k(E[\bm{Z}_1 \bm{Z}_1^{\top}])
\]

\begin{lemma}
\label{lemma:eigSShatL1}
Under the above assumptions, $\lambda_i(\bm{H}) = \Theta(n)$ with high probability when $i \le k$, i.e.\ the largest $k$ eigenvalues of $\bm{H}$ is of order $n$. Moreover, we have $\| \bm{S} \|_2 = \Theta(n)$ and $\| \hat{\bm{S}} \|_2 = \Theta(n)$ with high probability.
\end{lemma}
\textbf{Remark:} This is an extended version of Proposition 4.3 in \citep{sussman2014consistent}.

We ignore the proofs of the following results since they are similar to the proofs in \citep{lyzinski2017community}.

\begin{lemma}
\label{lemma:AlmostOrthogonalL1}
Let $\bm{W}_1 \Sigma \bm{W}_2^{\top}$ be the singular value decomposition of $\bm{U}^{\top} \hat{\bm{U}}$. Then for sufficiently large $n$,
\[
	\| \bm{U}^{\top} \hat{\bm{U}} - \bm{W}_1 \bm{W}_2^{\top} \|_F = O(m^{-1} n^{-1} \log n)
\]
with high probability.
\end{lemma}

We will denote the orthogonal matrix $\bm{W}_1 \bm{W}_2^{\top}$ by $\bm{W}^*$.

\begin{lemma}
\label{lemma:exchangeL1}
For sufficiently large $n$,
\[
	\| \bm{W}^* \hat{\bm{S}} - \bm{S} \bm{W}^* \|_F = O(m^{-1/2} \log n),
\]
\[
	\|\bm{W}^* \hat{\bm{S}}^{1/2} - \bm{S}^{1/2} \bm{W}^* \|_F = O(m^{-1/2} n^{-1/2} \log n)
\]
and
\[
	\| \bm{W}^* \hat{\bm{S}}^{-1/2} - \bm{S}^{-1/2} \bm{W}^* \|_F = O(m^{-1/2} n^{-3/2} \log n)
\]
with high probability.
\end{lemma}

\begin{lemma}
\label{lemma:XhatDiffXWexpressionL1}
There exists a rotation matrix $\bm{W}$ such that for sufficiently large $n$,
\[
	\|\hat{\bm{Z}} - \bm{Z} \bm{W}\|_F = \| (\hat{\bm{P}} - \bm{H}) \bm{U} \bm{S}^{-1/2} \|_F + O(m^{-1/2} n^{-1/2} (\log n)^{3/2})
\]
with high probability.
\end{lemma}

\begin{theorem}
\label{thm:XhatDiffXWL1}
There exists a rotation matrix $\bm{W}$ such that for sufficiently large $n$,
\[
	\max_i \| \hat{\bm{Z}}_i - \bm{W} \bm{Z}_i \|_2 = O(m^{-1/2} n^{-1/2} (\log n)^{3/2})
\]
with high probability.
\end{theorem}

\subsection{\texorpdfstring{$\widetilde{\bm{P}}^{(1)}$}{$P$} vs. \texorpdfstring{$\hat{\bm{P}}^{(1)}$}{$P$}}
\label{section:pf_MLEvsMLEASE2}

\begin{lemma}
\label{lemma:1stMomentPhatDiffL1}
$\left|  \hat{\bm{Z}}_i^{\top} \hat{\bm{Z}}_j - \bm{Z}_i^{\top} \bm{Z}_j \right| = O(m^{-1/2} n^{-1/2} (\log n)^{3/2})$ with high probability.
\end{lemma}
\begin{proof}
Let $\bm{W}$ be the rotation matrix in Theorem~\ref{thm:XhatDiffXWL1}, then
\begin{align*}
	\left|  \hat{\bm{Z}}_i^{\top} \hat{\bm{Z}}_j - \bm{Z}_i^{\top} \bm{Z}_j \right|
    = & \left| \hat{\bm{Z}}_i^{\top} \hat{\bm{Z}}_j - \hat{\bm{Z}}_i^{\top} \bm{W} \bm{Z}_j + \hat{\bm{Z}}_i^{\top} \bm{W} \bm{Z}_j - (\bm{W} \bm{Z}_i)^{\top} \bm{W} \bm{Z}_j \right| \\
    \le & \left| \hat{\bm{Z}}_i^{\top} (\hat{\bm{Z}}_j - \bm{W} \bm{Z}_j) + (\hat{\bm{Z}}_i^{\top} - (\bm{W} \bm{Z}_i)^{\top}) \bm{W} \bm{Z}_j \right| \\
    \le & \|\hat{\bm{Z}}_i\|_2 \|\hat{\bm{Z}}_j - \bm{W} \bm{Z}_j\|_2 + \|\bm{Z}_j\|_2 \|\hat{\bm{Z}}_i^{\top} - (\bm{W} \bm{Z}_i)^{\top}\|_2.
\end{align*}
Since $\|\bm{Z}_i\|_2^2 = \bm{Z}_i^{\top} \bm{Z}_i = \bm{H}^{(1)}_{ii} =  E[\hat{\bm{P}}^{(1)}_{ii}] = (1-\epsilon) \bm{P}_{ij} + \epsilon \bm{C}_{ij} \le R$, we have $\|\bm{Z}_i\|_2 = O(1)$.
Combined with Theorem~\ref{thm:XhatDiffXWL1},
\begin{align*}
    \left|  \hat{\bm{Z}}_i^{\top} \hat{\bm{Z}}_j - \bm{Z}_i^{\top} \bm{Z}_j \right|
    = & (\|\hat{\bm{Z}}_i\|_2 + \|\bm{Z}_j\|_2) O(m^{-1/2} n^{-1/2} (\log n)^{3/2}) \\
    \le & (\|\hat{\bm{Z}}_i - \bm{W} \bm{Z}_i\|_2 + \|\bm{W} \bm{Z}_i\|_2 + \|\bm{Z}_j\|_2) O(m^{-1/2} n^{-1/2} (\log n)^{3/2}) \\
    = & O(m^{-1/2} n^{-1/2} (\log n)^{3/2})
\end{align*}
with high probability.
\end{proof}

\begin{definition}
\label{def:truncationMLE}
Define $\widetilde{\bm{P}}_{ij}^{(1)} = (\hat{\bm{Z}}_i^{\top} \hat{\bm{Z}}_j)_{\mathrm{tr}}$, our estimator for $\bm{P}_{ij}$, to be a projection of $\hat{\bm{Z}}_i^{\top} \hat{\bm{Z}}_j$ onto $[0, \min(\hat{\bm{P}}_{ij}^{(1)}, R)]$.
\end{definition}

\begin{remark}
\label{remark:truncation}
The truncation step above to construct estimator is only for technical reasons. Since the constant $R$ could be arbitrarily large, we do not need this truncation step in practice. Note that Theorem~\ref{thm:MLEvsMLEASE} still holds with this modified estimator. And all our simulation and real data experiment do not contain this truncation procedure.
\end{remark}

\begin{lemma}[Theorem~\ref{thm:MLEvsMLEASE} Part 1]
\label{lm:L1Consistentproof}
Assuming that $m = O(n^b)$ for any $b > 0$, then the estimator based on ASE of MLE has the same entry-wise asymptotic bias as MLE, i.e.\
\[
	\lim_{n \to \infty} \mathrm{Bias}(\widetilde{\bm{P}}_{ij}^{(1)}) = \lim_{n \to \infty} E[\widetilde{\bm{P}}_{ij}^{(1)}] - \bm{P}_{ij} = \lim_{n \to \infty} E[\hat{\bm{P}}^{(1)}_{ij}] - \bm{P}_{ij}
    = \lim_{n \to \infty} \mathrm{Bias}(\hat{\bm{P}}_{ij}^{(1)}).
\]
\end{lemma}
\begin{proof}
Fix some $a > 0$, we have
\begin{align*}
	& E[|(\hat{\bm{Z}}_i^{\top} \hat{\bm{Z}}_j)_{\mathrm{tr}} - \bm{Z}_i^{\top} \bm{Z}_j|] \\
	= & E[|(\hat{\bm{Z}}_i^{\top} \hat{\bm{Z}}_j)_{\mathrm{tr}} - \bm{Z}_i^{\top} \bm{Z}_j| \mathbb{I}\{\hat{\bm{P}}_{ij} \le a\}]
	+ E[|(\hat{\bm{Z}}_i^{\top} \hat{\bm{Z}}_j)_{\mathrm{tr}} - \bm{Z}_i^{\top} \bm{Z}_j| \mathbb{I}\{\hat{\bm{P}}_{ij} > a\}]
\end{align*}
For the first term, we have
\begin{align*}
	& E[|(\hat{\bm{Z}}_i^{\top} \hat{\bm{Z}}_j)_{\mathrm{tr}} - \bm{Z}_i^{\top} \bm{Z}_j| \mathbb{I}\{\hat{\bm{P}}_{ij} \le a\}] \\
	\le & E[|(\hat{\bm{Z}}_i^{\top} \hat{\bm{Z}}_j)_{\mathrm{tr}} - \bm{Z}_i^{\top} \bm{Z}_j| \mathbb{I}\{\hat{\bm{P}}_{ij} \le a\} \mathbb{I}\{\mathrm{Lemma\ } \ref{lemma:1stMomentPhatDiffL1} \mathrm{\ holds}\}] \\
	& + E[|(\hat{\bm{Z}}_i^{\top} \hat{\bm{Z}}_j)_{\mathrm{tr}} - \bm{Z}_i^{\top} \bm{Z}_j| \mathbb{I}\{\hat{\bm{P}}_{ij} \le a\} \mathbb{I}\{\mathrm{Lemma\ } \ref{lemma:1stMomentPhatDiffL1} \mathrm{\ does\ not\ hold}\}] \\
	\le & E[|(\hat{\bm{Z}}_i^{\top} \hat{\bm{Z}}_j)_{\mathrm{tr}} - \bm{Z}_i^{\top} \bm{Z}_j| \mathbb{I}\{\hat{\bm{P}}_{ij} \le a\} | \mathrm{Lemma\ } \ref{lemma:1stMomentPhatDiffL1} \mathrm{\ holds}] \\
	& + n^{-c} E[|(\hat{\bm{Z}}_i^{\top} \hat{\bm{Z}}_j)_{\mathrm{tr}} - \bm{Z}_i^{\top} \bm{Z}_j| \mathbb{I}\{\hat{\bm{P}}_{ij} \le a\} | \mathrm{Lemma\ } \ref{lemma:1stMomentPhatDiffL1} \mathrm{\ does\ not\ hold}] \\
	\le & O(m^{-1/2} n^{-1/2} (\log n)^{3/2}) \\
	& + n^{-c} E[|(\hat{\bm{Z}}_i^{\top} \hat{\bm{Z}}_j)_{\mathrm{tr}} - \hat{\bm{P}}_{ij}| \mathbb{I}\{\hat{\bm{P}}_{ij} \le a\}| \mathrm{Lemma\ } \ref{lemma:1stMomentPhatDiffL1} \mathrm{\ does\ not\ hold}] \\
	& + n^{-c} E[|\hat{\bm{P}}_{ij} - \bm{Z}_i^{\top} \bm{Z}_j| \mathbb{I}\{\hat{\bm{P}}_{ij} \le a\}| \mathrm{Lemma\ } \ref{lemma:1stMomentPhatDiffL1} \mathrm{\ does\ not\ hold}] \\
	\le & O(m^{-1/2} n^{-1/2} (\log n)^{3/2}) + n^{-c} E[\hat{\bm{P}}_{ij} \mathbb{I}\{\hat{\bm{P}}_{ij} \le a\}| \mathrm{Lemma\ } \ref{lemma:1stMomentPhatDiffL1} \mathrm{\ does\ not\ hold}] \\
	& + n^{-c} E[(\hat{\bm{P}}_{ij} + R) \mathbb{I}\{\hat{\bm{P}}_{ij} \le a\}| \mathrm{Lemma\ } \ref{lemma:1stMomentPhatDiffL1} \mathrm{\ does\ not\ hold}] \\
	\le & O(m^{-1/2} n^{-1/2} (\log n)^{3/2}) + a n^{-c} + (a+R) n^{-c} \\
	\le & O(m^{-1/2} n^{-1/2} (\log n)^{3/2}) + 2 n^{-c} (a + R).
\end{align*}
Notice that
\begin{align*}
	& E[\hat{\bm{P}}_{ij} \mathbb{I} \{ \hat{\bm{P}}_{ij} > a \}]
	= E[\left(\frac{1}{m} \sum_{1 \le t \le m} \bm{A}_{ij}^{(t)}\right) \mathbb{I} \{ \hat{\bm{P}}_{ij} > a \}] \\
	= & \frac{1}{m} E[\sum_{1 \le t \le m} \bm{A}_{ij}^{(t)} \mathbb{I} \{ \hat{\bm{P}}_{ij} > a \}]
	\le \frac{1}{m} E[\sum_{1 \le t \le m} \bm{A}_{ij}^{(t)} \mathbb{I} \{ \max_{1 \le s \le m} \bm{A}_{ij}^{(s)} > a \}] \\
	\le & \frac{1}{m} E[\sum_{1 \le t \le m} \bm{A}_{ij}^{(t)} \left(\sum_{1 \le s \le m}\mathbb{I} \{ \bm{A}_{ij}^{(s)} > a \}\right)]
	= E[\bm{A}_{ij}^{(1)} \left(\sum_{1 \le s \le m}\mathbb{I} \{ \bm{A}_{ij}^{(s)} > a \}\right)] \\
	= & E[\bm{A}_{ij}^{(1)} \mathbb{I} \{ \bm{A}_{ij}^{(1)} > a \})] + (m-1) E[\bm{A}_{ij}^{(1)} \mathbb{I} \{ \bm{A}_{ij}^{(2)} > a \})] \\
	= & E[\bm{A}_{ij}^{(1)} \mathbb{I} \{ \bm{A}_{ij}^{(1)} > a \})] + (m-1) E[\bm{A}_{ij}^{(1)}] P(\bm{A}_{ij}^{(1)} > a),
\end{align*}
and similarly
\begin{align*}
	& E[(\hat{\bm{P}}_{ij} + R) \mathbb{I} \{ \hat{\bm{P}}_{ij} > a \}] \\
	= & E[\hat{\bm{P}}_{ij} \mathbb{I} \{ \hat{\bm{P}}_{ij} > a \}] + R \cdot P(\hat{\bm{P}}_{ij} > a) \\
	\le & E[\bm{A}_{ij}^{(1)} \mathbb{I} \{ \bm{A}_{ij}^{(1)} > a \})] + (m-1) E[\bm{A}_{ij}^{(1)}] P(\bm{A}_{ij}^{(1)} > a)
	+ R \cdot m \cdot P(\bm{A}_{ij}^{(1)} > a).
\end{align*}
Thus for the second term,
\begin{align*}
	& E[|(\hat{\bm{Z}}_i^{\top} \hat{\bm{Z}}_j)_{\mathrm{tr}} - \bm{Z}_i^{\top} \bm{Z}_j| \mathbb{I}\{\hat{\bm{P}}_{ij} > a\}] \\
	\le & E[|(\hat{\bm{Z}}_i^{\top} \hat{\bm{Z}}_j)_{\mathrm{tr}} - \hat{\bm{P}}_{ij}| \mathbb{I}\{\hat{\bm{P}}_{ij} > a\}] + E[|\hat{\bm{P}}_{ij} - \bm{Z}_i^{\top} \bm{Z}_j| \mathbb{I}\{\hat{\bm{P}}_{ij} > a\}] \\
	\le & E[\hat{\bm{P}}_{ij} \mathbb{I}\{\hat{\bm{P}}_{ij} > a\}] + E[(\hat{\bm{P}}_{ij} + R) \mathbb{I}\{\hat{\bm{P}}_{ij} > a\}] \\
	\le & 2 E[\bm{A}_{ij}^{(1)} \mathbb{I} \{ \bm{A}_{ij}^{(1)} > a \})] + 2(m-1) E[\bm{A}_{ij}^{(1)}] P(\bm{A}_{ij}^{(1)} > a) \\
	& + R \cdot m \cdot P(\bm{A}_{ij}^{(1)} > a) \\
	\le & 2 e^{-a/R} (a + 2 m R).
\end{align*}
Thus
\begin{align*}
	& E[|(\hat{\bm{Z}}_i^{\top} \hat{\bm{Z}}_j)_{\mathrm{tr}} - \bm{Z}_i^{\top} \bm{Z}_j|] \\
	\le & O(m^{-1/2} n^{-1/2} (\log n)^{3/2}) + 2 n^{-c} (a + R) + 2 e^{-a/R} (a + 2 m R).
\end{align*}
Let $a = m^{-1} n^{2b}$ for any $b > 0$, and $c = 2b + 3$, combined with the assumption $m = O(n^{b})$, we have
\begin{align*}
	& E[|(\hat{\bm{Z}}_i^{\top} \hat{\bm{Z}}_j)_{\mathrm{tr}} - \bm{Z}_i^{\top} \bm{Z}_j|] \\
	= & O(m^{-1/2} n^{-1/2} (\log n)^{3/2}) + O(m^{-1} n^{-3}) + O(m^{-1} n^{2b}) \cdot O(e^{-m^{-1} n^{2b}}) \\
	= & O(m^{-1/2} n^{-1/2} (\log n)^{3/2}) + O(m^{-1} n^{-3}) + O(m^{-1} n^{2b}) \cdot O(e^{- n^{b}}) \\
	= & O(m^{-1/2} n^{-1/2} (\log n)^{3/2}) + O(m^{-1} n^{-3}) + O(m^{-1} n^{2b}) \cdot O(n^{-2b-3}) \\
	= & O(m^{-1/2} n^{-1/2} (\log n)^{3/2}) + O(m^{-1} n^{-3}) \\
	= & O(m^{-1/2} n^{-1/2} (\log n)^{3/2}).
\end{align*}
\end{proof}

\begin{theorem}
\label{thm:VarASEL1proof}
Assuming that $m = O(n^b)$ for any $b > 0$, then $\mathrm{Var}((\hat{\bm{Z}}_i^{\top} \hat{\bm{Z}}_j)_{\mathrm{tr}}) = O(m^{-1} n^{-1} (\log n)^3)$.
\end{theorem}
\begin{proof}
By Lemma~\ref{lemma:1stMomentPhatDiffL1},
\begin{align*}
	\mathrm{Var}((\hat{\bm{Z}}_i^{\top} \hat{\bm{Z}}_j)_{\mathrm{tr}})
    = & E[((\hat{\bm{Z}}_i^{\top} \hat{\bm{Z}}_j)_{\mathrm{tr}} - E[(\hat{\bm{Z}}_i^{\top} \hat{\bm{Z}}_j)_{\mathrm{tr}}])^2] \\
    = & E[((\hat{\bm{Z}}_i^{\top} \hat{\bm{Z}}_j)_{\mathrm{tr}} - \bm{Z}_i^{\top} \bm{Z}_j + \bm{Z}_i^{\top} \bm{Z}_j - E[(\hat{\bm{Z}}_i^{\top} \hat{\bm{Z}}_j)_{\mathrm{tr}}])^2] \\
    = & E[((\hat{\bm{Z}}_i^{\top} \hat{\bm{Z}}_j)_{\mathrm{tr}} - \bm{Z}_i^{\top} \bm{Z}_j)^2] + E[(\bm{Z}_i^{\top} \bm{Z}_j - E[(\hat{\bm{Z}}_i^{\top} \hat{\bm{Z}}_j)_{\mathrm{tr}}])^2] \\
    & + 2E[((\hat{\bm{Z}}_i^{\top} \hat{\bm{Z}}_j)_{\mathrm{tr}} - \bm{Z}_i^{\top} \bm{Z}_j)(\bm{Z}_i^{\top} \bm{Z}_j - E[(\hat{\bm{Z}}_i^{\top} \hat{\bm{Z}}_j)_{\mathrm{tr}}])] \\
    \le & E[((\hat{\bm{Z}}_i^{\top} \hat{\bm{Z}}_j)_{\mathrm{tr}} - \bm{Z}_i^{\top} \bm{Z}_j)^2] + E[(\bm{Z}_i^{\top} \bm{Z}_j - E[(\hat{\bm{Z}}_i^{\top} \hat{\bm{Z}}_j)_{\mathrm{tr}}])^2] \\
    & + 2\sqrt{E[((\hat{\bm{Z}}_i^{\top} \hat{\bm{Z}}_j)_{\mathrm{tr}} - \bm{Z}_i^{\top} \bm{Z}_j)^2] E[(\bm{Z}_i^{\top} \bm{Z}_j - E[(\hat{\bm{Z}}_i^{\top} \hat{\bm{Z}}_j)_{\mathrm{tr}}])^2]} \\
    \le & 4 E[((\hat{\bm{Z}}_i^{\top} \hat{\bm{Z}}_j)_{\mathrm{tr}} - \bm{Z}_i^{\top} \bm{Z}_j)^2].
\end{align*}
Fix some $a > 0$, we have
\begin{align*}
	& E[((\hat{\bm{Z}}_i^{\top} \hat{\bm{Z}}_j)_{\mathrm{tr}} - \bm{Z}_i^{\top} \bm{Z}_j)^2] \\
	= & E[((\hat{\bm{Z}}_i^{\top} \hat{\bm{Z}}_j)_{\mathrm{tr}} - \bm{Z}_i^{\top} \bm{Z}_j)^2 \mathbb{I}\{\hat{\bm{P}}_{ij} \le a\}]
	+ E[((\hat{\bm{Z}}_i^{\top} \hat{\bm{Z}}_j)_{\mathrm{tr}} - \bm{Z}_i^{\top} \bm{Z}_j)^2 \mathbb{I}\{\hat{\bm{P}}_{ij} > a\}].
\end{align*}
For the first term, we have
\begin{align*}
	& E[((\hat{\bm{Z}}_i^{\top} \hat{\bm{Z}}_j)_{\mathrm{tr}} - \bm{Z}_i^{\top} \bm{Z}_j)^2 \mathbb{I}\{\hat{\bm{P}}_{ij} \le a\}] \\
	\le & E[((\hat{\bm{Z}}_i^{\top} \hat{\bm{Z}}_j)_{\mathrm{tr}} - \bm{Z}_i^{\top} \bm{Z}_j)^2 \mathbb{I}\{\hat{\bm{P}}_{ij} \le a\} \mathbb{I}\{\mathrm{Lemma\ } \ref{lemma:1stMomentPhatDiffL1} \mathrm{\ holds}\}] \\
	& + E[((\hat{\bm{Z}}_i^{\top} \hat{\bm{Z}}_j)_{\mathrm{tr}} - \bm{Z}_i^{\top} \bm{Z}_j)^2 \mathbb{I}\{\hat{\bm{P}}_{ij} \le a\} \mathbb{I}\{\mathrm{Lemma\ } \ref{lemma:1stMomentPhatDiffL1} \mathrm{\ does\ not\ hold}\}] \\
	\le & E[((\hat{\bm{Z}}_i^{\top} \hat{\bm{Z}}_j)_{\mathrm{tr}} - \bm{Z}_i^{\top} \bm{Z}_j)^2 \mathbb{I}\{\hat{\bm{P}}_{ij} \le a\} | \mathrm{Lemma\ } \ref{lemma:1stMomentPhatDiffL1} \mathrm{\ holds}] \\
	& + n^{-c} E[((\hat{\bm{Z}}_i^{\top} \hat{\bm{Z}}_j)_{\mathrm{tr}} - \bm{Z}_i^{\top} \bm{Z}_j)^2 \mathbb{I}\{\hat{\bm{P}}_{ij} \le a\} | \mathrm{Lemma\ } \ref{lemma:1stMomentPhatDiffL1} \mathrm{\ does\ not\ hold}] \\
	\le & O(m^{-1} n^{-1} (\log n)^3) \\
	& + 2 n^{-c} E[((\hat{\bm{Z}}_i^{\top} \hat{\bm{Z}}_j)_{\mathrm{tr}} - \hat{\bm{P}}_{ij})^2 \mathbb{I}\{\hat{\bm{P}}_{ij} \le a\}| \mathrm{Lemma\ } \ref{lemma:1stMomentPhatDiffL1} \mathrm{\ does\ not\ hold}] \\
	& + 2 n^{-c} E[(\hat{\bm{P}}_{ij} - \bm{Z}_i^{\top} \bm{Z}_j)^2 \mathbb{I}\{\hat{\bm{P}}_{ij} \le a\}| \mathrm{Lemma\ } \ref{lemma:1stMomentPhatDiffL1} \mathrm{\ does\ not\ hold}] \\
	\le & O(m^{-1} n^{-1} (\log n)^3) + 2 n^{-c} E[\hat{\bm{P}}_{ij}^2 \mathbb{I}\{\hat{\bm{P}}_{ij} \le a\}| \mathrm{Lemma\ } \ref{lemma:1stMomentPhatDiffL1} \mathrm{\ does\ not\ hold}] \\
	& + 2 n^{-c} E[(\hat{\bm{P}}_{ij} + R)^2 \mathbb{I}\{\hat{\bm{P}}_{ij} \le a\}| \mathrm{Lemma\ } \ref{lemma:1stMomentPhatDiffL1} \mathrm{\ does\ not\ hold}] \\
	\le & O(m^{-1} n^{-1} (\log n)^3) + 2 a^2 n^{-c} + 2 (a+R)^2 n^{-c} \\
	\le & O(m^{-1} n^{-1} (\log n)^3) + 4 n^{-c} (a + R)^2.
\end{align*}
Notice that
\begin{align*}
	& E[\hat{\bm{P}}_{ij}^2 \mathbb{I} \{ \hat{\bm{P}}_{ij} > a \}]
	= E[(\frac{1}{m} \sum_{1 \le t \le m} \bm{A}_{ij}^{(t)})^2 \mathbb{I} \{ \hat{\bm{P}}_{ij} > a \}] \\
	\le & \frac{1}{m} E[\sum_{1 \le t \le m} \bm{A}_{ij}^{(t)2} \mathbb{I} \{ \hat{\bm{P}}_{ij} > a \}]
	\le \frac{1}{m} E[\sum_{1 \le t \le m} \bm{A}_{ij}^{(t)2} \mathbb{I} \{ \max_{1 \le s \le m} \bm{A}_{ij}^{(s)} > a \}] \\
	\le & \frac{1}{m} E[\sum_{1 \le t \le m} \bm{A}_{ij}^{(t)2} (\sum_{1 \le s \le m}\mathbb{I} \{ \bm{A}_{ij}^{(s)} > a \})]
	= E[\bm{A}_{ij}^{(1)2} (\sum_{1 \le s \le m}\mathbb{I} \{ \bm{A}_{ij}^{(s)} > a \})] \\
	= & E[\bm{A}_{ij}^{(1)2} \mathbb{I} \{ \bm{A}_{ij}^{(1)} > a \})] + (m-1) E[\bm{A}_{ij}^{(1)2} \mathbb{I} \{ \bm{A}_{ij}^{(2)} > a \})] \\
	= & E[\bm{A}_{ij}^{(1)2} \mathbb{I} \{ \bm{A}_{ij}^{(1)} > a \})] + (m-1) E[\bm{A}_{ij}^{(1)2}] P(\bm{A}_{ij}^{(1)} > a),
\end{align*}
and similarly
\begin{align*}
	& E[(\hat{\bm{P}}_{ij} + R)^2 \mathbb{I} \{ \hat{\bm{P}}_{ij} > a \}] \\
	= & E[\hat{\bm{P}}_{ij}^2 \mathbb{I} \{ \hat{\bm{P}}_{ij} > a \}] + 2R \cdot E[\hat{\bm{P}}_{ij} \mathbb{I} \{ \hat{\bm{P}}_{ij} > a \}] + R^2 P(\hat{\bm{P}}_{ij} > a) \\
	\le & E[\bm{A}_{ij}^{(1)2} \mathbb{I} \{ \bm{A}_{ij}^{(1)} > a \})] + (m-1) E[\bm{A}_{ij}^{(1)2}] P(\bm{A}_{ij}^{(1)} > a) \\
	& + 2R \left( E[\bm{A}_{ij}^{(1)} \mathbb{I} \{ \bm{A}_{ij}^{(1)} > a \})] + (m-1) E[\bm{A}_{ij}^{(1)}] P(\bm{A}_{ij}^{(1)} > a) \right) \\
	& + R^2 \cdot m \cdot P(\bm{A}_{ij}^{(1)} > a).
\end{align*}
Thus for the second term,
\begin{align*}
	& E[((\hat{\bm{Z}}_i^{\top} \hat{\bm{Z}}_j)_{\mathrm{tr}} - \bm{Z}_i^{\top} \bm{Z}_j)^2 \mathbb{I}\{\hat{\bm{P}}_{ij} > a\}] \\
	\le & 2 E[((\hat{\bm{Z}}_i^{\top} \hat{\bm{Z}}_j)_{\mathrm{tr}} - \hat{\bm{P}}_{ij})^2 \mathbb{I}\{\hat{\bm{P}}_{ij} > a\}] + 2 E[(\hat{\bm{P}}_{ij} - \bm{Z}_i^{\top} \bm{Z}_j)^2 \mathbb{I}\{\hat{\bm{P}}_{ij} > a\}] \\
	\le & 2 E[\hat{\bm{P}}_{ij}^2 \mathbb{I}\{\hat{\bm{P}}_{ij} > a\}] + 2 E[(\hat{\bm{P}}_{ij} + R)^2 \mathbb{I}\{\hat{\bm{P}}_{ij} > a\}] \\
	\le & 4 E[\bm{A}_{ij}^{(1)2} \mathbb{I} \{ \bm{A}_{ij}^{(1)} > a \})] + 4(m-1) E[\bm{A}_{ij}^{(1)2}] P(\bm{A}_{ij}^{(1)} > a) \\
	& + 4R \cdot E[\bm{A}_{ij}^{(1)} \mathbb{I} \{ \bm{A}_{ij}^{(1)} > a \})] + 2R(m-1) E[\bm{A}_{ij}^{(1)}] P(\bm{A}_{ij}^{(1)} > a) \\
	& + 2R^2 \cdot m \cdot P(\bm{A}_{ij}^{(1)} > a) \\
	\le & 4 e^{-a/R} \left( a^2 + 3R a + 3(m+1)R^2 \right) \\
	\le & 4 e^{-a/R} (a + 2 m^{1/2} R)^2.
\end{align*}

Thus,
\[
	\mathrm{Var}((\hat{\bm{Z}}_i^{\top} \hat{\bm{Z}}_j)_{\mathrm{tr}}) \le
	O(m^{-1} n^{-1} (\log n)^3) + 16 (a+R)^2 n^{-c} + 16 (a+2 m^{1/2} R)^2 e^{-a/R}.
\]
Let $a = m^{-1/2} n^{b}$ for any $b > 0$, and $c = 2b + 3$, combined with the assumption $m = O(n^{b})$, we have
\begin{align*}
	\mathrm{Var}((\hat{\bm{Z}}_i^{\top} \hat{\bm{Z}}_j)_{\mathrm{tr}})
	= & O(m^{-1} n^{-1} (\log n)^3) + O(m^{-1} n^{-3}) + O(m^{-1} n^{2b}) \cdot O(e^{-m^{-1/2} n^{b}}) \\
	= & O(m^{-1} n^{-1} (\log n)^3) + O(m^{-1} n^{-3}) + O(m^{-1} n^{2b}) \cdot O(e^{- n^{b/2}}) \\
	= & O(m^{-1} n^{-1} (\log n)^3) + O(m^{-1} n^{-3}) + O(m^{-1} n^{2b}) \cdot O(n^{-2b-3}) \\
	= & O(m^{-1} n^{-1} (\log n)^3) + O(m^{-1} n^{-3}) \\
	= & O(m^{-1} n^{-1} (\log n)^3).
\end{align*}
\end{proof}

\begin{theorem}[Theorem~\ref{thm:MLEvsMLEASE} Part 2]
\label{thm:AREL1proof}
Assuming that $m = O(n^b)$ for any $b > 0$,  then for $1 \le i, j \le n$ and $i \ne j$,
\[
	\frac{\mathrm{Var}(\widetilde{\bm{P}}_{ij}^{(1)})}{\mathrm{Var}(\hat{\bm{P}}_{ij}^{(1)})}
    = O(n^{-1} (\log n)^3).
\]
And thus
\[
	\mathrm{ARE}(\hat{\bm{P}}_{ij}^{(1)}, \widetilde{\bm{P}}_{ij}^{(1)}) = 0.
\]
\end{theorem}
\begin{proof}
The results are direct from Theorem~\ref{thm:VarASEL1proof} and Theorem~\ref{thm:MLEvsMLqE}.
\end{proof}

\subsection{\texorpdfstring{$\widetilde{\bm{P}}^{(q)}$}{$P$} vs. \texorpdfstring{$\hat{\bm{P}}^{(q)}$}{$P$}}
\label{section:pf_MLqEASEvsMLqE}

\begin{theorem}
\label{thm:PqDiff}
Let $\bm{P}$ and $\bm{C}$ be two $n$-by-$n$ symmetric and hollow matrices satisfying element-wise conditions $0 < \bm{P}_{ij} \le \bm{C}_{ij} \le R$ for some constant $R > 0$. For $0 < \epsilon < 1$, we define $m$ symmetric and hollow matrices as
\[
	\bm{A}^{(t)} \stackrel{iid}{\sim} (1-\epsilon) \mathrm{Exp}(\bm{P}) + \epsilon \mathrm{Exp}(\bm{C})
\]
for $1 \le t \le m$.
Let $\hat{\bm{P}}^{(q)}$ be the entry-wise ML$q$E based on exponential distribution with $m$ observations.
Define $\bm{H}^{(q)} = E[\hat{\bm{P}}^{(q)}]$,
then for any constant $c > 0$ there exists another constant $n_0(c)$, independent of $n$, $\bm{P}$, $\bm{C}$ and $\epsilon$, such that if $n > n_0$, then for all $\eta$ satisfying $n^{-c} \le \eta \le 1/2$,
\[
	P \left( \| \hat{\bm{P}}^{(q)} - \bm{H}^{(q)} \|_2 \le 8 R \sqrt{2 n \ln(n/\eta)}) \right) \ge 1 - \eta.
\]
\end{theorem}
\begin{proof}
Similar to the proof of Theorem \ref{thm:P1Diff}.

By Lemma \ref{lemma:LqlMLE} we have
\begin{align*}
	\left| \hat{\bm{P}}^{(q)}_{ij} - \bm{H}^{(q)}_{ij} \right|
    = & \left| \hat{\bm{P}}^{(q)}_{ij} - \hat{\bm{P}}^{(1)}_{ij} + \hat{\bm{P}}^{(1)}_{ij} - \bm{H}^{(1)}_{ij} + \bm{H}^{(1)}_{ij} - \bm{H}^{(q)}_{ij} \right| \\
    \le & \left| \hat{\bm{P}}^{(q)}_{ij} - \hat{\bm{P}}^{(1)}_{ij} \right| + \left| \hat{\bm{P}}^{(1)}_{ij} - \bm{H}^{(1)}_{ij} \right| + \left| \bm{H}^{(1)}_{ij} - \bm{H}^{(q)}_{ij} \right| \\
    \le & \hat{\bm{P}}^{(1)}_{ij} + \left| \hat{\bm{P}}^{(1)}_{ij} - \bm{H}^{(1)}_{ij} \right| +  \bm{H}^{(1)}_{ij} \\
    \le & 2 \left( \left| \hat{\bm{P}}^{(1)}_{ij} - \bm{H}^{(1)}_{ij} \right| + \bm{H}^{(1)}_{ij} \right).
\end{align*}
Also,
\begin{align*}
	E[(\hat{\bm{P}}^{(q)}_{ij} - \bm{H}^{(q)}_{ij})^k]
    \le & E\left[\left|\hat{\bm{P}}^{(q)}_{ij} - \bm{H}^{(q)}_{ij}\right|^k\right] \\
    \le & 2^k E \left[ \left( \left| \hat{\bm{P}}^{(1)}_{ij} - \bm{H}^{(1)}_{ij} \right| + \bm{H}^{(1)}_{ij} \right)^k \right] \\
    \le & 2^k \sum_{s = 0}^k \binom{k}{s} E \left[ \left| \hat{\bm{P}}^{(1)}_{ij} - \bm{H}^{(1)}_{ij} \right|^s \right] \left( \bm{H}^{(1)}_{ij} \right)^{k-s} \\
    \le & 2^k \sum_{s = 0}^k \binom{k}{s} R^s s! \left( \bm{H}^{(1)}_{ij} \right)^{k-s} \\
    \le & 2^k k! \sum_{s = 0}^k \binom{k}{s} R^s \left( \bm{H}^{(1)}_{ij} \right)^{k-s} \\
    = & 2^k k! \left( R + \bm{H}^{(1)}_{ij} \right)^k \\
    \le & 2^{2k} k! R^k.
    \stepcounter{equation}\tag{\theequation}\label{eqn:expectpdiffqpowerk}
\end{align*}
Therefore we have
\[
	P \left( \| \hat{\bm{P}}^{(q)} - \bm{H}^{(q)} \| \ge t \right)
    \le n \exp \left( - \frac{t^2/2}{32 R^2 n + R t} \right).
\]

Now let $c > 0$ be given and assume $n^{-c} \le \eta \le 1/2$. Then there exists a $n_0(c)$ independent of $n$, $\bm{P}$, $\bm{C}$ and $\epsilon$ such that whenever $n > n_0(c)$,
\[
	t = 8 R \sqrt{2 n \ln(n/\eta)} \le 32 R n.
\]

Plugging this $t$ into the equation above, we get
\[
	P(\| \hat{\bm{P}}^{(q)} - \bm{H}^{(q)} \| \ge 8 R \sqrt{2 n \ln(n/\eta)})
    \le n \exp\left(-\frac{t^2}{64 R^2 n}\right) = \eta.
\]
\end{proof}

As we define $\bm{H}^{(q)} = E[\hat{\bm{P}}^{(q)}]$, let $d^{(q)} = \mathrm{rank}(\bm{H}^{(q)})$ be the dimension in which we are going to embed $\hat{\bm{P}}^{(q)}$. Notice that it is less than or equal to $K\times K'$ based on the SBM assumption. Then we can define $\bm{H}^{(q)} = \bm{Z} \bm{Z}^{\top}$ where $\bm{Z} \in \mathbb{R}^{n \times d^{(q)}}$.

For simplicity, from now on, we will use $\hat{\bm{P}}$ to represent $\hat{\bm{P}}^{(q)}$, use $\bm{H}$ to represent $\bm{H}^{(q)}$ and use $k$ to represent the dimension $d^{(q)}$ we are going to embed. Assume $\bm{H} = \bm{U} \bm{S} \bm{U}^{\top} = \bm{Z} \bm{Z}^{\top}$, where $\bm{Z} = [\bm{Z}_1, \cdots, \bm{Z}_n]^{\top}$ is a $n$-by-$k$ matrix. Then our estimate for $\bm{Z}$ up to rotation is $\hat{\bm{Z}} = \hat{\bm{U}} \hat{\bm{S}}^{1/2}$, where $\hat{\bm{U}} \hat{\bm{S}} \hat{\bm{U}}^{\top}$ is the rank-$d$ spectral decomposition of $|\hat{\bm{P}}| = (\hat{\bm{P}}^{\top} \hat{\bm{P}})^{1/2}$.

Furthermore, we assume that the second moment matrix $E[\bm{Z}_1 \bm{Z}_1^{\top}]$ is rank $k$ and has distinct eigenvalues $\lambda_i(E[\bm{Z}_1 \bm{Z}_1^{\top}])$. In particular, we assume that there exists $\delta > 0$ such that
\[
	\delta < \lambda_k(E[\bm{Z}_1 \bm{Z}_1^{\top}])
\]

\begin{lemma}
\label{lemma:eigSShat}
Under the above assumptions, $\lambda_i(\bm{H}) = \Theta(n)$ with high probability when $i \le k$, i.e.\ the largest $k$ eigenvalues of $\bm{H}$ is of order $n$. Moreover, we have $\| \bm{S} \|_2 = \Theta(n)$ and $\| \hat{\bm{S}} \|_2 = \Theta(n)$ with high probability.
\end{lemma}
\begin{proof}
Exactly the same as proof for Lemma \ref{lemma:eigSShatL1}.
\end{proof}

\begin{lemma}
\label{lemma:AlmostOrthogonal}
Let $\bm{W}_1 \Sigma \bm{W}_2^{\top}$ be the singular value decomposition of $\bm{U}^{\top} \hat{\bm{U}}$. Then for sufficiently large $n$,
\[
	\| \bm{U}^{\top} \hat{\bm{U}} - \bm{W}_1 \bm{W}_2^{\top} \|_F = O(n^{-1} \log n)
\]
with high probability.
\end{lemma}
\begin{proof}
Exactly the same as proof for Lemma \ref{lemma:AlmostOrthogonalL1}.
\end{proof}

We will denote the orthogonal matrix $\bm{W}_1 \bm{W}_2^{\top}$ by $\bm{W}^*$.
\begin{lemma}
\label{lemma:exchange}
For sufficiently large $n$,
\[
	\| \bm{W}^* \hat{\bm{S}} - \bm{S} \bm{W}^* \|_F = O(\log n),
\]
\[
	\|\bm{W}^* \hat{\bm{S}}^{1/2} - \bm{S}^{1/2} \bm{W}^* \|_F = O(n^{-1/2} \log n)
\]
and
\[
	\| \bm{W}^* \hat{\bm{S}}^{-1/2} - \bm{S}^{-1/2} \bm{W}^* \|_F = O(n^{-3/2} \log n)
\]
with high probability.
\end{lemma}
\begin{proof}
Similar to the proof of Lemma \ref{lemma:exchangeL1}.
\end{proof}

\begin{lemma}
\label{lemma:XhatDiffXWexpression}
There exists a rotation matrix $\bm{W}$ such that for sufficiently large $n$,
\[
	\|\hat{\bm{Z}} - \bm{Z} \bm{W}\|_F = \| (\hat{\bm{P}} - \bm{H}) \bm{U} \bm{S}^{-1/2} \|_F + O(n^{-1/2} (\log n)^{3/2})
\]
with high probability.
\end{lemma}
\begin{proof}
Exactly the same as proof for Lemma \ref{lemma:XhatDiffXWexpressionL1}.
\end{proof}

\begin{theorem}
\label{thm:XhatDiffXW}
There exists a rotation matrix $\bm{W}$ such that for sufficiently large $n$,
\[
	\max_i \| \hat{\bm{Z}}_i - \bm{W} \bm{Z}_i \|_2 = O(n^{-1/2} (\log n)^{3/2})
\]
with high probability.
\end{theorem}
\begin{proof}
Similar to the proof of Theorem \ref{thm:XhatDiffXWL1}.
\end{proof}

\begin{lemma}
\label{lemma:1stMomentPhatDiffLq}
$\left|  \hat{\bm{Z}}_i^{\top} \hat{\bm{Z}}_j - \bm{Z}_i^{\top} \bm{Z}_j \right| = O(n^{-1/2} (\log n)^{3/2})$ with high probability.
\end{lemma}
\begin{proof}
Similar to the proof of Lemma \ref{lemma:1stMomentPhatDiffL1}.
\end{proof}

\begin{definition}
Define $\widetilde{\bm{P}}_{ij}^{(q)} = (\hat{\bm{Z}}_i^{\top} \hat{\bm{Z}}_j)_{\mathrm{tr}}$, our estimator for $\bm{P}_{ij}$, to be a projection of $\hat{\bm{Z}}_i^{\top} \hat{\bm{Z}}_j$ onto $[0, \min(\hat{\bm{P}}_{ij}^{(q)}, R)]$.
\end{definition}

\begin{lemma}[Theorem~\ref{thm:MLqEvsMLqEASE} Part 1]
\label{lm:LqConsistentproof}
Assuming that $m = O(n^b)$ for any $b > 0$, then the estimator based on ASE of ML$q$E has the same entry-wise asymptotic bias as ML$q$E, i.e.\
\[
	\lim_{n \to \infty} \mathrm{Bias}(\widetilde{\bm{P}}_{ij}^{(q)}) = \lim_{n \to \infty} E[\widetilde{\bm{P}}_{ij}^{(q)}] - \bm{P}_{ij} = \lim_{n \to \infty} E[\hat{\bm{P}}^{(q)}_{ij}] - \bm{P}_{ij}
    = \lim_{n \to \infty} \mathrm{Bias}(\hat{\bm{P}}_{ij}^{(q)}).
\]
\end{lemma}
\begin{proof}
Fix some $a > 0$, we have
\begin{align*}
	& E[|(\hat{\bm{Z}}_i^{\top} \hat{\bm{Z}}_j)_{\mathrm{tr}} - \bm{Z}_i^{\top} \bm{Z}_j|] \\
	= & E[|(\hat{\bm{Z}}_i^{\top} \hat{\bm{Z}}_j)_{\mathrm{tr}} - \bm{Z}_i^{\top} \bm{Z}_j| \mathbb{I}\{\hat{\bm{P}}^{(1)}_{ij} \le a\}]
	+ E[|(\hat{\bm{Z}}_i^{\top} \hat{\bm{Z}}_j)_{\mathrm{tr}} - \bm{Z}_i^{\top} \bm{Z}_j| \mathbb{I}\{\hat{\bm{P}}^{(1)}_{ij} > a\}].
\end{align*}
Note that we are thresholding according to $\hat{\bm{P}}^{(1)}$ instead of $\hat{\bm{P}}^{(q)}$. By Lemma \ref{lemma:LqlMLE}, we know $\hat{\bm{P}}^{(q)} < \hat{\bm{P}}^{(1)}$ given any data.
For the first term, we have
\begin{align*}
	& E[|(\hat{\bm{Z}}_i^{\top} \hat{\bm{Z}}_j)_{\mathrm{tr}} - \bm{Z}_i^{\top} \bm{Z}_j| \mathbb{I}\{\hat{\bm{P}}^{(1)}_{ij} \le a\}] \\
	\le & E[|(\hat{\bm{Z}}_i^{\top} \hat{\bm{Z}}_j)_{\mathrm{tr}} - \bm{Z}_i^{\top} \bm{Z}_j| \mathbb{I}\{\hat{\bm{P}}^{(1)}_{ij} \le a\} \mathbb{I}\{\mathrm{Lemma\ } \ref{lemma:1stMomentPhatDiffLq} \mathrm{\ holds}\}] \\
	& + E[|(\hat{\bm{Z}}_i^{\top} \hat{\bm{Z}}_j)_{\mathrm{tr}} - \bm{Z}_i^{\top} \bm{Z}_j| \mathbb{I}\{\hat{\bm{P}}^{(1)}_{ij} \le a\} \mathbb{I}\{\mathrm{Lemma\ } \ref{lemma:1stMomentPhatDiffLq} \mathrm{\ does\ not\ hold}\}] \\
	\le & E[(\hat{\bm{Z}}_i^{\top} \hat{\bm{Z}}_j)_{\mathrm{tr}} - \bm{Z}_i^{\top} \bm{Z}_j| \mathbb{I}\{\hat{\bm{P}}^{(1)}_{ij} \le a\} |\{\mathrm{Lemma\ } \ref{lemma:1stMomentPhatDiffLq} \mathrm{\ holds}] \\
	& + n^{-c} E[|(\hat{\bm{Z}}_i^{\top} \hat{\bm{Z}}_j)_{\mathrm{tr}} - \bm{Z}_i^{\top} \bm{Z}_j| \mathbb{I}\{\hat{\bm{P}}^{(1)}_{ij} \le a\} | \mathrm{Lemma\ } \ref{lemma:1stMomentPhatDiffLq} \mathrm{\ does\ not\ hold}] \\
	\le & O(n^{-1/2} (\log n)^{3/2}) \\
	& + n^{-c} E[|(\hat{\bm{Z}}_i^{\top} \hat{\bm{Z}}_j)_{\mathrm{tr}} - \hat{\bm{P}}^{(q)}_{ij}| \mathbb{I}\{\hat{\bm{P}}^{(1)}_{ij} \le a\} | \mathrm{Lemma\ } \ref{lemma:1stMomentPhatDiffLq} \mathrm{\ does\ not\ hold}] \\
	& + n^{-c} E[|\hat{\bm{P}}^{(q)}_{ij} - \bm{Z}_i^{\top} \bm{Z}_j| \mathbb{I}\{\hat{\bm{P}}^{(1)}_{ij} \le a\} | \mathrm{Lemma\ } \ref{lemma:1stMomentPhatDiffLq} \mathrm{\ does\ not\ hold}] \\
	\le & O(n^{-1/2} (\log n)^{3/2}) \\
	& + n^{-c} E[\hat{\bm{P}}^{(q)}_{ij} \mathbb{I}\{\hat{\bm{P}}^{(1)}_{ij} \le a\} | \mathrm{Lemma\ } \ref{lemma:1stMomentPhatDiffLq} \mathrm{\ does\ not\ hold}] \\
	& + n^{-c} E[(\hat{\bm{P}}^{(q)}_{ij} + R) \mathbb{I}\{\hat{\bm{P}}^{(1)}_{ij} \le a\} | \mathrm{Lemma\ } \ref{lemma:1stMomentPhatDiffLq} \mathrm{\ does\ not\ hold}] \\
	\le & O(n^{-1/2} (\log n)^{3/2}) \\
	& + n^{-c} E[\hat{\bm{P}}^{(1)}_{ij} \mathbb{I}\{\hat{\bm{P}}^{(1)}_{ij} \le a\} | \mathrm{Lemma\ } \ref{lemma:1stMomentPhatDiffLq} \mathrm{\ does\ not\ hold}] \\
	& + n^{-c} E[(\hat{\bm{P}}^{(1)}_{ij} + R) \mathbb{I}\{\hat{\bm{P}}^{(1)}_{ij} \le a\} | \mathrm{Lemma\ } \ref{lemma:1stMomentPhatDiffLq} \mathrm{\ does\ not\ hold}] \\
	\le & O(n^{-1/2} (\log n)^{3/2}) + a n^{-c} + (a+R) n^{-c} \\
	\le & O(n^{-1/2} (\log n)^{3/2}) + 2 n^{-c} (a + R).
\end{align*}
For the second term, we have
\begin{align*}
	& E[|(\hat{\bm{Z}}_i^{\top} \hat{\bm{Z}}_j)_{\mathrm{tr}} - \bm{Z}_i^{\top} \bm{Z}_j| \mathbb{I}\{\hat{\bm{P}}^{(1)}_{ij} > a\}] \\
	\le & E[|(\hat{\bm{Z}}_i^{\top} \hat{\bm{Z}}_j)_{\mathrm{tr}} - \hat{\bm{P}}^{(q)}_{ij}| \mathbb{I}\{\hat{\bm{P}}^{(1)}_{ij} > a\}] + E[|\hat{\bm{P}}^{(q)}_{ij} - \bm{Z}_i^{\top} \bm{Z}_j| \mathbb{I}\{\hat{\bm{P}}^{(1)}_{ij} > a\}] \\
	\le & E[\hat{\bm{P}}^{(q)}_{ij} \mathbb{I}\{\hat{\bm{P}}^{(1)}_{ij} > a\}] + E[(\hat{\bm{P}}^{(q)}_{ij} + R) \mathbb{I}\{\hat{\bm{P}}^{(1)}_{ij} > a\}] \\
	\le & E[\hat{\bm{P}}^{(1)}_{ij} \mathbb{I}\{\hat{\bm{P}}^{(1)}_{ij} > a\}] + E[(\hat{\bm{P}}^{(1)}_{ij} + R) \mathbb{I}\{\hat{\bm{P}}^{(1)}_{ij} > a\}] \\
	\le & 2 e^{-a/R} (a + 2 m R).
\end{align*}
Similarly, assuming $m = O(n^b)$ for any $b > 0$, we have
\[
E[|(\hat{\bm{Z}}_i^{\top} \hat{\bm{Z}}_j)_{\mathrm{tr}} - \bm{Z}_i^{\top} \bm{Z}_j|] = O(n^{-1/2} (\log n)^{3/2}).
\]
\end{proof}

\begin{theorem}
\label{thm:VarASELqproof}
Assuming that $m = O(n^b)$ for any $b > 0$, then $\mathrm{Var}((\hat{\bm{Z}}_i^{\top} \hat{\bm{Z}}_j)_{\mathrm{tr}}) = O(n^{-1} (\log n)^3)$.
\end{theorem}
\begin{proof}
By Lemma \ref{lemma:1stMomentPhatDiffLq},
\begin{align*}
	\mathrm{Var}((\hat{\bm{Z}}_i^{\top} \hat{\bm{Z}}_j)_{\mathrm{tr}})
    = & E[((\hat{\bm{Z}}_i^{\top} \hat{\bm{Z}}_j)_{\mathrm{tr}} - E[(\hat{\bm{Z}}_i^{\top} \hat{\bm{Z}}_j)_{\mathrm{tr}}])^2] \\
    = & E[((\hat{\bm{Z}}_i^{\top} \hat{\bm{Z}}_j)_{\mathrm{tr}} - \bm{Z}_i^{\top} \bm{Z}_j + \bm{Z}_i^{\top} \bm{Z}_j - E[(\hat{\bm{Z}}_i^{\top} \hat{\bm{Z}}_j)_{\mathrm{tr}}])^2] \\
    = & E[((\hat{\bm{Z}}_i^{\top} \hat{\bm{Z}}_j)_{\mathrm{tr}} - \bm{Z}_i^{\top} \bm{Z}_j)^2] + E[(\bm{Z}_i^{\top} \bm{Z}_j - E[(\hat{\bm{Z}}_i^{\top} \hat{\bm{Z}}_j)_{\mathrm{tr}}])^2] \\
    & + 2E[((\hat{\bm{Z}}_i^{\top} \hat{\bm{Z}}_j)_{\mathrm{tr}} - \bm{Z}_i^{\top} \bm{Z}_j)(\bm{Z}_i^{\top} \bm{Z}_j - E[(\hat{\bm{Z}}_i^{\top} \hat{\bm{Z}}_j)_{\mathrm{tr}}])] \\
    \le & E[((\hat{\bm{Z}}_i^{\top} \hat{\bm{Z}}_j)_{\mathrm{tr}} - \bm{Z}_i^{\top} \bm{Z}_j)^2] + E[(\bm{Z}_i^{\top} \bm{Z}_j - E[(\hat{\bm{Z}}_i^{\top} \hat{\bm{Z}}_j)_{\mathrm{tr}}])^2] \\
    & + 2\sqrt{E[((\hat{\bm{Z}}_i^{\top} \hat{\bm{Z}}_j)_{\mathrm{tr}} - \bm{Z}_i^{\top} \bm{Z}_j)^2] E[(\bm{Z}_i^{\top} \bm{Z}_j - E[(\hat{\bm{Z}}_i^{\top} \hat{\bm{Z}}_j)_{\mathrm{tr}}])^2]} \\
    \le & 4 E[((\hat{\bm{Z}}_i^{\top} \hat{\bm{Z}}_j)_{\mathrm{tr}} - \bm{Z}_i^{\top} \bm{Z}_j)^2].
\end{align*}
Fix some $a > 0$, we have
\begin{align*}
	& E[((\hat{\bm{Z}}_i^{\top} \hat{\bm{Z}}_j)_{\mathrm{tr}} - \bm{Z}_i^{\top} \bm{Z}_j)^2] \\
	= & E[((\hat{\bm{Z}}_i^{\top} \hat{\bm{Z}}_j)_{\mathrm{tr}} - \bm{Z}_i^{\top} \bm{Z}_j)^2 \mathbb{I}\{\hat{\bm{P}}^{(1)}_{ij} \le a\}]
	+ E[((\hat{\bm{Z}}_i^{\top} \hat{\bm{Z}}_j)_{\mathrm{tr}} - \bm{Z}_i^{\top} \bm{Z}_j)^2 \mathbb{I}\{\hat{\bm{P}}^{(1)}_{ij} > a\}].
\end{align*}
Note that we are thresholding according to $\hat{\bm{P}}^{(1)}$ instead of $\hat{\bm{P}}^{(q)}$. By Lemma \ref{lemma:LqlMLE}, we know $\hat{\bm{P}}^{(q)} < \hat{\bm{P}}^{(1)}$ given any data.
For the first term, we have
\begin{align*}
	& E[((\hat{\bm{Z}}_i^{\top} \hat{\bm{Z}}_j)_{\mathrm{tr}} - \bm{Z}_i^{\top} \bm{Z}_j)^2 \mathbb{I}\{\hat{\bm{P}}^{(1)}_{ij} \le a\}] \\
	\le & E[((\hat{\bm{Z}}_i^{\top} \hat{\bm{Z}}_j)_{\mathrm{tr}} - \bm{Z}_i^{\top} \bm{Z}_j)^2 \mathbb{I}\{\hat{\bm{P}}^{(1)}_{ij} \le a\} \mathbb{I}\{\mathrm{Lemma\ } \ref{lemma:1stMomentPhatDiffLq} \mathrm{\ holds}\}] \\
	& + E[((\hat{\bm{Z}}_i^{\top} \hat{\bm{Z}}_j)_{\mathrm{tr}} - \bm{Z}_i^{\top} \bm{Z}_j)^2 \mathbb{I}\{\hat{\bm{P}}^{(1)}_{ij} \le a\} \mathbb{I}\{\mathrm{Lemma\ } \ref{lemma:1stMomentPhatDiffLq} \mathrm{\ does\ not\ hold}\}] \\
	\le & E[((\hat{\bm{Z}}_i^{\top} \hat{\bm{Z}}_j)_{\mathrm{tr}} - \bm{Z}_i^{\top} \bm{Z}_j)^2 \mathbb{I}\{\hat{\bm{P}}^{(1)}_{ij} \le a\} |\{\mathrm{Lemma\ } \ref{lemma:1stMomentPhatDiffLq} \mathrm{\ holds}] \\
	& + n^{-c} E[((\hat{\bm{Z}}_i^{\top} \hat{\bm{Z}}_j)_{\mathrm{tr}} - \bm{Z}_i^{\top} \bm{Z}_j)^2 \mathbb{I}\{\hat{\bm{P}}^{(1)}_{ij} \le a\} | \mathrm{Lemma\ } \ref{lemma:1stMomentPhatDiffLq} \mathrm{\ does\ not\ hold}] \\
	\le & O(n^{-1} (\log n)^3) \\
	& + 2 n^{-c} E[((\hat{\bm{Z}}_i^{\top} \hat{\bm{Z}}_j)_{\mathrm{tr}} - \hat{\bm{P}}^{(q)}_{ij})^2 \mathbb{I}\{\hat{\bm{P}}^{(1)}_{ij} \le a\} | \mathrm{Lemma\ } \ref{lemma:1stMomentPhatDiffLq} \mathrm{\ does\ not\ hold}] \\
	& + 2 n^{-c} E[(\hat{\bm{P}}^{(q)}_{ij} - \bm{Z}_i^{\top} \bm{Z}_j)^2 \mathbb{I}\{\hat{\bm{P}}^{(1)}_{ij} \le a\} | \mathrm{Lemma\ } \ref{lemma:1stMomentPhatDiffLq} \mathrm{\ does\ not\ hold}] \\
	\le & O(n^{-1} (\log n)^3) \\
	& + 2 n^{-c} E[\hat{\bm{P}}^{(q)2}_{ij} \mathbb{I}\{\hat{\bm{P}}^{(1)}_{ij} \le a\} | \mathrm{Lemma\ } \ref{lemma:1stMomentPhatDiffLq} \mathrm{\ does\ not\ hold}] \\
	& + 2 n^{-c} E[(\hat{\bm{P}}^{(q)}_{ij} + R)^2 \mathbb{I}\{\hat{\bm{P}}^{(1)}_{ij} \le a\} | \mathrm{Lemma\ } \ref{lemma:1stMomentPhatDiffLq} \mathrm{\ does\ not\ hold}] \\
	\le & O(n^{-1} (\log n)^3)
	+ 2 n^{-c} E[\hat{\bm{P}}^{(1)2}_{ij} \mathbb{I}\{\hat{\bm{P}}^{(1)}_{ij} \le a\} | \mathrm{Lemma\ } \ref{lemma:1stMomentPhatDiffLq} \mathrm{\ does\ not\ hold}] \\
	& + 2 n^{-c} E[(\hat{\bm{P}}^{(1)}_{ij} + R)^2 \mathbb{I}\{\hat{\bm{P}}^{(1)}_{ij} \le a\} | \mathrm{Lemma\ } \ref{lemma:1stMomentPhatDiffLq} \mathrm{\ does\ not\ hold}] \\
	\le & O(n^{-1} (\log n)^3) + 2 a^2 n^{-c} + 2 (a+R)^2 n^{-c} \\
	\le & O(n^{-1} (\log n)^3) + 4 n^{-c} (a + R)^2.
\end{align*}
For the second term, we have
\begin{align*}
	& E[((\hat{\bm{Z}}_i^{\top} \hat{\bm{Z}}_j)_{\mathrm{tr}} - \bm{Z}_i^{\top} \bm{Z}_j)^2 \mathbb{I}\{\hat{\bm{P}}^{(1)}_{ij} > a\}] \\
	\le & 2 E[((\hat{\bm{Z}}_i^{\top} \hat{\bm{Z}}_j)_{\mathrm{tr}} - \hat{\bm{P}}^{(q)}_{ij})^2 \mathbb{I}\{\hat{\bm{P}}^{(1)}_{ij} > a\}] + 2 E[(\hat{\bm{P}}^{(q)}_{ij} - \bm{Z}_i^{\top} \bm{Z}_j)^2 \mathbb{I}\{\hat{\bm{P}}^{(1)}_{ij} > a\}] \\
	\le & 2 E[\hat{\bm{P}}^{(q)2}_{ij} \mathbb{I}\{\hat{\bm{P}}^{(1)}_{ij} > a\}] + 2 E[(\hat{\bm{P}}^{(q)}_{ij} + R)^2 \mathbb{I}\{\hat{\bm{P}}^{(1)}_{ij} > a\}] \\
	\le & 2 E[\hat{\bm{P}}^{(1)2}_{ij} \mathbb{I}\{\hat{\bm{P}}^{(1)}_{ij} > a\}] + 2 E[(\hat{\bm{P}}^{(1)}_{ij} + R)^2 \mathbb{I}\{\hat{\bm{P}}^{(1)}_{ij} > a\}] \\
	\le & 4 e^{-a/R} (a + 2 m^{1/2} R)^2.
\end{align*}
Similarly, assuming $m = O(n^b)$ for any $b > 0$, we have
\[
	\mathrm{Var}((\hat{\bm{Z}}_i^{\top} \hat{\bm{Z}}_j)_{\mathrm{tr}})
	= O(n^{-1} (\log n)^3).
\]
\end{proof}

\begin{theorem}[Theorem~\ref{thm:MLqEvsMLqEASE} Part 2]
\label{thm:ARELqproof}
Assuming that $m = O(n^b)$ for any $b > 0$,  then for $1 \le i, j \le n$ and $i \ne j$,
\[
	\frac{\mathrm{Var}(\widetilde{\bm{P}}_{ij}^{(q)})}{\mathrm{Var}(\hat{\bm{P}}_{ij}^{(q)})}
    = O(m n^{-1} (\log n)^3).
\]
Moreover, if $m = o(n (\log n)^{-3})$, then
\[
	\mathrm{ARE}(\hat{\bm{P}}_{ij}^{(q)}, \widetilde{\bm{P}}_{ij}^{(q)}) = 0.
\]
\end{theorem}
\begin{proof}
The results are direct from Theorem~\ref{thm:VarASELqproof} and Theorem~\ref{thm:MLEvsMLqE}.
\end{proof}

\subsection{\texorpdfstring{$\widetilde{\bm{P}}^{(q)}$}{$P$} vs. \texorpdfstring{$\widetilde{\bm{P}}^{(1)}$}{$P$}}
\label{section:MLqEASEvsMLEASE}
\begin{theorem}
\label{thm:biasL1andLq}
For sufficiently large values of $\{\bm{C}_{ij}\}$ and any $1 \le i,j \le n$, if $m \to \infty$ at order $m = O(n^b)$ for any $b > 0$, then the estimator based on ASE of ML$q$E has smaller entry-wise asymptotic bias compared to the estimator based on ASE of MLE, i.e.\
\[
	\lim_{m, n \to \infty} \mathrm{Bias}(\widetilde{\bm{P}}_{ij}^{(1)})
    > \lim_{m, n \to \infty} \mathrm{Bias}(\widetilde{\bm{P}}_{ij}^{(q)})
\]
\end{theorem}
\begin{proof}
Direct result from Theorem~\ref{thm:MLEvsMLqE}, Theorem~\ref{thm:MLEvsMLEASE} and Theorem~\ref{thm:MLqEvsMLqEASE}.
\end{proof}

\begin{theorem}
\label{thm:varianceL1andLq}
For any $1 \le i,j \le n$, if $m = O(n^b)$ for any $b > 0$, then
\[
	\lim_{n \to \infty} \mathrm{Var}(\widetilde{\bm{P}}_{ij}^{(1)})
    = \lim_{n \to \infty} \mathrm{Var}(\widetilde{\bm{P}}_{ij}^{(q)}) = 0.
\]
\end{theorem}
\begin{proof}
Direct result from Theorem~\ref{thm:MLEvsMLEASE} and Theorem~\ref{thm:MLqEvsMLqEASE}.
\end{proof}

\subsection{Other Proofs}
\label{section:pf_other}

\begin{lemma}
\label{lm:poisson}
Let $\bm{A}_{ij} \stackrel{ind}{\sim} (1-\epsilon) f_{\bm{P}_{ij}} + \epsilon f_{\bm{C}_{ij}}$ with $f$ to be Poisson, then $E[(\bm{A}_{ij} - E[\hat{\bm{P}}_{ij}^{(1)}])^k] \le \mathrm{const}^k \cdot k!$, where $\hat{\bm{P}}^{(1)}$ is the entry-wise MLE as defined before.
\end{lemma}
\begin{proof}
First we prove $(x - \theta)^k \le k! (\exp(x-\theta) + \exp(\theta - x))$.
\begin{enumerate}
\item $k$ is even. Then by Taylor expansion, $\exp(x - \theta) + \exp(\theta - x) \ge {(x-\theta)^k}/{(k!)}$
\item $k$ is odd. When $x \ge \theta$, still by Taylor expansion, $(x-\theta)^k \le k! \cdot \exp(x-\theta)$. When $x < \theta$, $(x-\theta)^k < 0 \le k! \cdot \exp(x-\theta)$.
\end{enumerate}
Thus $(x - \theta)^k \le k! \cdot (\exp(x-\theta) + \exp(\theta - x))$.
So the $k$-th central moment of Poisson distribution with parameter $\theta$ is bounded by
\begin{align*}
E[(X-\theta)^k] & \le k! \left( E[e^{X-\theta}] + E[e^{\theta - X}] \right) \\
& = k! \left( e^{-\theta} E[e^X] + e^{\theta} E[e^{-X}] \right) \\
& = k! \left( e^{\theta(e - 2)} + e^{\theta e^{-1}} \right).
\end{align*}
Let $X_1 \sim \mathrm{Poisson}(\bm{P}_{ij})$ and $X_2 \sim \mathrm{Poisson}(\bm{C}_{ij})$.
Then if $\bm{A}_{ij}$ is distributed from a mixture model as in the statement, we have
\begin{align*}
& E[(\bm{A}_{ij} - E[\hat{\bm{P}}_{ij}^{(1)}])^k] \\
= & (1-\epsilon) E[(X_1 - \bm{P}_{ij} + \bm{P}_{ij} - E[\hat{\bm{P}}_{ij}^{(1)}])] +
\epsilon E[(X_2 - \bm{C}_{ij} + \bm{C}_{ij} - E[\hat{\bm{P}}_{ij}^{(1)}])] \\
= & (1-\epsilon) \sum_{j = 0}^k \binom{k}{j} (\bm{P}_{ij} - E[\hat{\bm{P}}_{ij}^{(1)}])^{k - j} E[(X_1 - \bm{P}_{ij})^j] \\
& + \epsilon \sum_{j = 0}^k \binom{k}{j} (\bm{C}_{ij} - E[\hat{\bm{P}}_{ij}^{(1)}])^{k - j} E[(X_2 - \bm{C}_{ij})^j] \\
\le & (1-\epsilon) \sum_{j = 0}^k \binom{k}{j} (\bm{P}_{ij} - E[\hat{\bm{P}}_{ij}^{(1)}])^{k - j} \cdot j! \cdot \mathrm{const} \\
& + \epsilon \sum_{j = 0}^k \binom{k}{j} (\bm{C}_{ij} - E[\hat{\bm{P}}_{ij}^{(1)}])^{k - j} \cdot j! \cdot \mathrm{const} \\
\le & (1-\epsilon) k! \cdot \mathrm{const}^k + \epsilon k! \cdot \mathrm{const}^k \\
\le & \mathrm{const}^k \cdot k!.
\end{align*}
\end{proof}

\end{document}